\documentclass[aps,pra,twocolumn,nofootinbib, superscriptaddress,10pt]{revtex4-2}

\newif\ifapp

\usepackage{silence}
\usepackage{float}

\usepackage{mathdots}
\usepackage{amsfonts,amssymb,amsmath}
\usepackage{amsthm}
\usepackage[]{graphics,graphicx,epsfig}
\usepackage[utf8]{inputenc}
\usepackage[T1]{fontenc}
\usepackage{epstopdf}
\usepackage[toc,page,header]{appendix}

\usepackage{bm}
\usepackage{bbm}
\usepackage{braket}
\usepackage{cases}
\usepackage{dcolumn}
\usepackage{hyperref}

\usepackage{mathtools,mleftright}
\mleftright
\def\<{\langle}  %% overriding the original command \<
\def\>{\rangle}  %% overriding the original command \>

\newtheorem{theorem}{Theorem}
\newtheorem{proposition}{Proposition}
\newtheorem{lemma}{Lemma}

\newcommand{\lsp}{\hspace{0.1em}}
\newcommand{\equad}{\,\hphantom{=}\,}

\newcommand{\diag}{\operatorname{diag}}
\newcommand{\tr}{\operatorname{tr}}

\newcommand{\lb}{\mathrm{LB}}
\newcommand{\ub}{\mathrm{UB}}
\newcommand{\lc}{\mathrm{LC}}
\newcommand{\sep}{\mathrm{sep}}
\newcommand{\opt}{\mathrm{opt}}
\newcommand{\MUB}{\mathrm{MUB}}
\newcommand{\SN}{\mathrm{SN}}
\newcommand{\SR}{\mathrm{SR}}

\newcommand{\rmi}{\mathrm{i}}

\newcommand{\bbD}{\mathbb{D}}
\newcommand{\bbE}{\mathbb{E}}

\newcommand{\bbR}{\mathbb{R}}
\newcommand{\bbone}{\mathbbm{1}}

\newcommand{\bfx}{\mathbf{x}}
\newcommand{\bfy}{\mathbf{y}}
\newcommand{\bfs}{\mathbf{s}}

\newcommand{\bmw}{\bm{w}}

\newcommand{\caB}{\mathcal{B}}
\newcommand{\caD}{\mathcal{D}}
\newcommand{\caE}{\mathcal{E}}
\newcommand{\caS}{\mathcal{S}}
\newcommand{\caH}{\mathcal{H}}

\newcommand{\rme}{\mathrm{e}}
\newcommand{\rmH}{\mathrm{H}}
\newcommand{\rmU}{\mathrm{U}}
\newcommand{\rmL}{\mathrm{L}}

\newcommand{\tp}{\tilde{p}}

\newcommand{\trmU}{\tilde{\rmU}}

\newcommand{\tcaS}{\tilde{\caS}}

\newcommand{\tvarphi}{\tilde{\varphi}}
\newcommand{\tpsi}{\tilde{\psi}}

\newcommand{\tPsi}{\tilde{\Psi}}

\newcommand{\ser}{\caS_{\caE_r}}
\newcommand{\serone}{\caS_{\caE_1}}
\newcommand{\tser}{\tcaS_{\caE_r}}

\newcommand{\eref}[1]{Eq.~\eqref{#1}}
\newcommand{\eqsref}[2]{Eqs.~\eqref{#1} and \eqref{#2}}
\newcommand{\Eref}[1]{Equation~\eqref{#1}}

\newcommand{\rcite}[1]{Ref.~\cite{#1}}
\newcommand{\rscite}[1]{Refs.~\cite{#1}}

\newcommand{\thref}[1]{Theorem~\ref{#1}}
\newcommand{\thsref}[1]{Theorems~\ref{#1}}

\newcommand{\Thsref}[1]{Theorems~\ref{#1}}

\newcommand{\pref}[1]{Proposition~\ref{#1}}
\newcommand{\psref}[1]{Propositions~\ref{#1}}
\newcommand{\Pref}[1]{Proposition~\ref{#1}}
\newcommand{\Psref}[1]{Propositions~\ref{#1}}

\newcommand{\lref}[1]{Lemma~\ref{#1}}
\newcommand{\lsref}[1]{Lemmas~\ref{#1}}

\newcommand{\Lsref}[1]{Lemmas~\ref{#1}}

\newcommand{\sref}[1]{Sec.~\ref{#1}}

\newcommand{\Sref}[1]{Section~\ref{#1}}

\newcommand{\aref}[1]{Appendix~\ref{#1}}

\newcommand{\fref}[1]{Fig.~\ref{#1}}

\newcommand{\Fref}[1]{Figure~\ref{#1}}

\hypersetup{colorlinks=true, citecolor=blue, urlcolor=blue, linkcolor=blue}

\begin{document}

\title{Efficient certification of high-dimensional entanglement}

\author{Yiwen Wu}
\author{Zihao Li}
\author{Huangjun Zhu}
\email{zhuhuangjun@fudan.edu.cn}

\affiliation{State Key Laboratory of Surface Physics, Department of Physics, and Center for Field Theory and Particle Physics, Fudan University, Shanghai 200433, China}

\affiliation{Institute for Nanoelectronic Devices and Quantum Computing, Fudan University, Shanghai 200433, China}

\affiliation{Shanghai Research Center for Quantum Sciences, Shanghai 201315, China}

\date{\today}

\begin{abstract}
High-dimensional entanglement (HDE) is a valuable resource in quantum information processing, and efficient certification of HDE is crucial to many applications. In this work, we propose a simple and general framework for certifying HDE in general bipartite pure states under restricted operations, such as local operations and classical communication (LOCC). On this basis we show that HDE in general bipartite pure states can be certified efficiently.  Moreover, the sample cost for certifying a given degree of HDE even decreases monotonically with the local dimensions. In addition, for a general two-qubit pure state, we  construct an optimal entanglement certification strategy based on separable operations, which can be realized by LOCC when the target state has sufficiently high entanglement. The core concept of our framework is versatile and can be extended to certify a wide range of critical resources under restricted operations.
\end{abstract}

\maketitle

\section{Introduction}

Quantum entanglement is a characteristic of quantum mechanics and a valuable resource in many tasks in quantum information processing \cite{Horo09}.
\emph{High-dimensional entanglement} (HDE) \cite{TerhalH00,ErhardKZ20}, which involves quantum systems with higher local dimensions compared with qubits, may further offer enhanced information capacity and improved noise resilience and has thus attracted increasing attention recently.  It is especially valuable for important tasks such as quantum cryptography \cite{Zhang13,Huber13}, quantum communication \cite{Cozzolino19,Ecker19,Hu21}, and quantum computation \cite{Lanyon09,Van13,Wang20}. Moreover, 
HDE is tied to the classical hardness in simulating quantum systems and may even serve as a benchmark for quantum technologies \cite{Amico08,Verstraete08,Eisert10}.  Therefore, efficient certification of HDE is of intrinsic interest to both theoretical studies and practical applications \cite{Guhne09,Friis19}. However, traditional tomographic approaches are too resource-intensive for this task for large and intermediate quantum systems.

To address the challenge in HDE certification, a number of alternative approaches have emerged recently, including fidelity-based Schmidt number witnesses \cite{Guo18,Bavaresco18,Chen20,Valencia20,Euler23,Li25}, correlation-based approaches \cite{Huang16,Erker17,Guo20,Wyderka23,Liu23,Lib25}, entropic steering criteria \cite{Schneeloch18,Dabrowski18}, and methods based on positive maps (such as the reduction map) \cite{Sanpera01,Mallick25,Yi25} and hypothesis testing \cite{Hu21ODHDE}. In addition, several verification protocols have been demonstrated in experiments across various platforms, including photonic systems \cite{Huang16,Bavaresco18,Guo18,Guo20,Valencia20,Chen20,Hu21ODHDE,Lib25} and cold atoms \cite{Dabrowski18,Euler23}. Despite this progress, little is known about the sample complexity of certifying HDE and the construction of optimal or nearly optimal certification protocols.

In this work, 
inspired by the idea of  \emph{quantum state verification} (QSV) \cite{Hayashi06,PLM18,ZhuHEVPQSshort19,ZhuHEVPQSlong19,Morris22,Yu22}, we propose a simple and general framework for certifying HDE under restricted operations, such as  \emph{local operations and classical communication} (LOCC) and  separable operations. To rigorously quantify the capabilities of these operations in certifying HDE, we introduce the concept of separation probabilities. Then,  we determine the separation probabilities of maximally entangled states and derive nearly tight upper and lower bounds for general bipartite pure states. Based on these findings, we show that HDE in general bipartite pure states can be certified efficiently using LOCC. Notably, the sample cost for certifying a given degree of HDE even decreases monotonically with the local dimensions. The same conclusion still holds even if the strategy is required to be homogeneous \cite{ZhuHEVPQSshort19,ZhuHEVPQSlong19}. In addition, we construct an optimal entanglement certification strategy for any two-qubit pure state using separable operations, and show that this optimal strategy can be realized by LOCC when the target state has sufficiently high entanglement. This study also shows  that optimal strategies for entanglement certification are in general different from the counterparts for QSV. The versatility of our framework extends beyond HDE certification. The basic idea may find applications in certifying many other important resources, such as coherence and nonstabilizerness, under restricted operations.

The rest of this paper is organized as follows.
In \sref{sec:Pre} we introduce necessary preliminaries on HDE and QSV. In \sref{sec:ResCert} we propose a simple and general framework for certifying entanglement under restricted operations and introduce the concept of separation probabilities. In \sref{sec:CertHDE} we show that   HDE in general bipartite pure states can be certified efficiently after clarifying the properties of separation probabilities. In  \sref{sec:CertHDEin2qPS}
we construct an optimal entanglement certification strategy for a general two-qubit pure state based on separable operations and show that this  strategy can be realized by LOCC when the target state has a high concurrence. \Sref{sec:Summary} summarizes this paper.

\section{\label{sec:Pre}Preliminaries}

Let $\caH$ be the Hilbert space of a quantum system under consideration  and denote by $\caD(\caH)$ the set of all quantum states on $\caH$. Given any pure state  $|\Psi\>$ in $\caH$ we use  $\Psi=|\Psi\>\<\Psi|$ to denote the corresponding density operator. Given a  closed subset  $\caS$ in $\caD(\caH)$, denote by 
$F(\Psi, \caS)$ the maximum fidelity between $|\Psi\>$ and states in $\caS$, that is,
\begin{equation}
	F(\Psi, \caS)=\max_{\sigma\in \caS} \<\Psi|\sigma|\Psi\>. 
\end{equation}
Denote by  $\tcaS$  the subset of pure states in $\caS$, that is,
\begin{equation}
	\tcaS:=\left\{\sigma\in \caS \mid \tr\left(\sigma^2\right)=1 \right\}, 
\end{equation}
which  is also a closed subset in $\caD(\caH)$. Given a positive integer $k$, let $[k]$ be a shorthand for $\{1,2,\ldots, k\}$.

\subsection{\label{ssec:HDE}Schmidt number and high-dimensional entanglement}

Here we assume that  $\caH=\caH_{AB}=\caH_A\otimes \caH_B$ is the Hilbert space of a bipartite quantum system shared by Alice and Bob; let $d_A=\dim(\caH_A)$, $d_B=\dim(\caH_B)$, $D=\dim(\caH_{AB})=d_A d_B$, and $d=\min\{d_A, d_B\}$. 
To simplify the following discussion, in this paper  we shall assume that $d=d_A\leq d_B$ without loss of generality. In addition, denote by $\rmU(\caH)$ the group of unitary operators on $\caH$ and by $\rmU(d)$ the group of unitary operators on a $d$-dimensional Hilbert space.

A bipartite pure state  $|\Psi\>\in \caH_{AB}$ is a product state if it can be expressed in the form $|\Psi\>=|\psi_A\>\otimes |\psi_B\>$ with $|\psi_A\>\in \caH_A$ and $|\psi_B\>\in \caH_B$. Otherwise, $|\Psi\>$ is entangled. The state $|\Psi\>$ is maximally entangled if the reduced state $\rho_A=\tr_B(|\Psi\>\<\Psi|)$ is a completely mixed state, in which case any other state in $\caD(\caH_{AB})$ can be generated from $|\Psi\>$ under LOCC (see \sref{ssec:LocalTran}). For example, 
here is a typical maximally entangled state: 
\begin{equation}\label{eq:MES}
	|\Phi\>=\sum_{j=0}^{d-1} \frac{1}{\sqrt{d}}|jj\>.
\end{equation}
A mixed state on $\caH_{AB}$ is separable if it is a convex combination of pure product states and entangled otherwise~\cite{Werner89}. A mixed state is maximally entangled if every pure state in its support is maximally entangled. The structures of such states were clarified in \rscite{Li12,Zhu21}.
The set of separable states on $\caH_{AB}$ is denoted by $\caS_\sep(\caH_{AB})$ henceforth, which can be abbreviated as $\caS_\sep$ if there is no danger of confusion.

Any bipartite pure state $|\Psi\>$ in $\caH_{AB}$ has a \emph{Schmidt decomposition} \cite{Peres95} of the form \begin{equation}\label{eq:SchmidtDecomposition}
	|\Psi\>=\sum_{j=0}^{d-1}\sqrt{s_j} |\psi_j^A\> \otimes |\psi_j^B\>,
\end{equation}
where $\{|\psi_j^A\>\}_{j=0}^{d-1}$ forms an orthonormal basis of $\caH_A$, $\{|\psi_j^B\>\}_{j=0}^{d-1}$ is a set of orthonormal states in $\caH_B$, and $\{s_j\}_{j=0}^{d-1}$ is the set of \emph{Schmidt coefficients}, also known as \emph{Schmidt spectrum}, which satisfies $\sum_{j=0}^{d-1}s_j=1$. Without loss of generality, we assume that $s_0\geq s_1 \geq \cdots \geq s_{d-1} \geq 0$ throughout this paper. Although the above decomposition is not necessarily unique, the set of Schmidt coefficients is uniquely determined by $|\Psi\>$. For example, 
all the Schmidt coefficients  are equal to $1/d$ whenever $|\Psi\>$ is maximally entangled, and vice versa. The Schmidt vector of $|\Psi\>$ is defined as $\bfs_{\Psi}:=(s_0,s_1,\ldots,s_{d-1})$. The \emph{Schmidt rank} of $|\Psi\>$ is defined as the number of nonzero Schmidt coefficients and is denoted by $\SR(\Psi)$ henceforth; it is equal to the rank of $\rho_A=\tr_B(|\Psi\>\<\Psi|)$ and also the rank of $\rho_B=\tr_A(|\Psi\>\<\Psi|)$.
 
Given $r\in [d]$, let  $\caS_r$ be the subset of quantum states in  $\caD(\caH_{AB})$ that can be expressed as convex combinations of pure states with Schmidt rank at most $r$. Note that $\caS_r$ is a proper subset of $\caS_{r+1}$ for $r\in [d-1]$, and $\caS_1$ coincides with the set $\caS_\sep$ of separable states. The \emph{Schmidt number} of $\sigma\in \caD(\caH_{AB})$ is defined as the smallest integer $r$ such that $\caS_r$ contains $\sigma$ and is denoted by $\SN(\sigma)$. By definition, the Schmidt number of a pure state is equal to its Schmidt rank.

As the simplest measure for quantifying HDE, 
the Schmidt number is discrete and not so robust to perturbation or noise. To remedy this problem,  a family of continuous entanglement measures were introduced by Vidal \cite{Vidal99}.
Suppose $|\Psi\>$ has Schmidt spectrum $\{s_j\}_{j=0}^{d-1}$. Define $	\caE_r(\Psi)$ as the sum of the $d-r$ smallest Schmidt coefficients, that is,
\begin{equation}\label{eq:Er}
	\caE_r(\Psi):=\sum_{j=r}^{d-1} s_j,\quad r\in [d-1].
\end{equation}
Note that $\caE_r(\Psi)$ only depends on the nonzero Schmidt coefficients of $|\Psi\>$ and is independent of the local dimensions (once the set of nonzero Schmidt coefficients is fixed).  The above  definition can  be extended to mixed states  via the convex-roof construction. An ensemble of pure states $\{|\Psi_l\>,\lsp  p_l\}_l$ is a convex decomposition of $\sigma\in \caD(\caH_{AB})$ if $\sum_l p_l|\Psi_l\>\<\Psi_l|=\sigma$. Denote by $\bbD(\sigma)$ the set of all convex decompositions of $\sigma$ into pure states. Then $\caE_r(\sigma)$ can be defined as follows:
\begin{equation}
	\caE_r(\sigma):= \inf_{\{|\Psi_l\>,\lsp  p_l\}_l\in \bbD(\sigma)}\sum_{l}p_l \caE_r(\Psi_l),
\end{equation}
where the infimum is taken over all convex decompositions of $\sigma$. Based on this definition we can introduce a subset of quantum states with limited HDE:
\begin{equation}
\ser(E):=\left\{ \sigma\in \caD(\caH_{AB}) \mid \caE_r(\sigma)\leq E \right\},
\end{equation}
which will be useful for formulating robust HDE certification. Note that $\ser(0)$ coincides with $\caS_r$.

Next, we clarify the basic properties of the entanglement measure $\caE_r$ that are relevant to the current study. \Psref{pro:ErLUB}-\ref{pro:ErLip} below are proved in \aref{app:ErProofs}, although \pref{pro:FidelityUB} is known before
\cite{TerhalH00,Horo99}.

\begin{proposition}\label{pro:ErLUB}
Suppose  $r\in[d-1]$; then 
\begin{equation}\label{eq:ErLUB}
	0\leq  \caE_r(\sigma)\leq \frac{d-r}{d} \quad \forall \, \sigma\in \caD(\caH_{AB}),
\end{equation}
where the lower bound is saturated iff $\SN(\sigma)\leq r$, and the upper bound is saturated iff $\sigma$ is maximally entangled. 
\end{proposition}
Note that \pref{pro:ErLUB} is compatible with the fact that  $\ser(0)=\caS_r$. Next, we 
clarify  the maximum fidelity between two bipartite pure states with given Schmidt spectra.
\begin{proposition}\label{pro:FidelityUB}
	Suppose $|\Psi\>$ and $|\Upsilon\>$ are two pure states in $\caH_{AB}$ that have Schmidt spectra $\{s_j\}_{j=0}^{d-1}$ and $\{t_j\}_{j=0}^{d-1}$, respectively. Then
	\begin{align}\label{eq:FidelityUB}
		|\<\Psi|\Upsilon\>|\leq \sum_{j=0}^{d-1}\sqrt{s_j t_j},
	\end{align}
	and the inequality is saturated if $|\Psi\>=\sum_{j=0}^{d-1}\sqrt{s_j} |jj\>$ and $|\Upsilon\>=\sum_{j=0}^{d-1}\sqrt{t_j} |jj\>$. 
\end{proposition}
As a simple corollary of \pref{pro:FidelityUB} we can deduce that
\begin{equation}
	|\<\Psi|\Phi\>| \leq \frac{1}{\sqrt{d}}\sum_{j=0}^{d-1}\sqrt{s_j},
\end{equation} 
and the inequality is saturated if $|\Psi\>=\sum_{j=0}^{d-1}\sqrt{s_j} |jj\>$. 
By virtue of \pref{pro:FidelityUB} we can derive the following proposition.
\begin{proposition}\label{pro:FPsiSk}
	Suppose $|\Psi\>\in \caH_{AB}$ has Schmidt spectrum $\{s_j\}_{j=0}^{d-1}$, $r\in[d-1]$,  and $E'=\caE_r(\Psi)$. Then 
	\begin{align}
		&F(\Psi,\caS_r)=F\left(\Psi,\tcaS_r\right)=1-\caE_r(\Psi), \label{eq:FPsiSk}  \\
		&F(\Psi,\ser(E))=F\left(\Psi,\tser(E)\right)\nonumber\\
		&=\begin{cases}
			\left[\sqrt{E' E}+\sqrt{(1-E')(1-E)}\lsp\right]^2 &\! 0\leq E< E',\\[0.5ex]
			1 &\! E\geq E'.
		\end{cases}\label{eq:FPsiSkE}
	\end{align}
	In addition, $F(\Psi,\ser(E))$ and $F(\Psi,\tser(E))$ are  nondecreasing and concave in $E$. 	
\end{proposition}
Note that the maximum fidelities $F(\Psi,\caS_r)=F(\Psi,\tcaS_r)$ and $F(\Psi,\ser(E))=F(\Psi,\tser(E))$ are independent of the local dimensions once the set of  nonzero Schmidt coefficients of $|\Psi\>$ is fixed. \Pref{pro:FPsiSk} endows $\caE_r(\Psi)$ with a simple operational interpretation, which will be discussed further in \sref{ssec:LocalTran}. 
In addition, $F(\Psi,\caS_1)=F(\Psi,\caS_\sep)=s_0$ coincides with the largest Schmidt coefficient of $|\Psi\>$ and is closely tied to the geometric measure of entanglement~\cite{Wei03}.

\begin{proposition}\label{pro:ErLip}
	Suppose  $|\Psi\>, |\Upsilon\>\in \caH_{AB}$ and $r\in[d-1]$. Then 
\begin{align}
|\caE_r(\Psi)-\caE_r(\Upsilon)|&\leq \ell(r,d)\sqrt{2-2|\<\Psi|\Upsilon\>|} \nonumber\\
&\leq \ell(r,d) \||\Psi\>-|\Upsilon\>\|_2, \label{eq:ErLip}
\end{align}
where
\begin{align}
\ell(r,d):=\begin{cases}
1 & r\leq d/2,\\
\frac{2\sqrt{r(d-r)}}{d} & r>d/2.
\end{cases}
\end{align}
\end{proposition}
\Pref{pro:ErLip} shows that $\caE_r(\Psi)$ is a  Lipschitz function with Lipschitz constant $\ell(r,d)\leq 1$.

\subsection{\label{ssec:LocalTran}Transformations of bipartite pure states under LOCC}

To understand the transformations of bipartite pure states under LOCC, we first need to introduce the concept of \emph{majorization} \cite{Bhatia96MA}.
Suppose  $\bfx=\left(x_0,x_1, \ldots, x_{d-1}\right)$ and $\bfy=\left(y_0,y_1, \ldots, y_{d-1}\right)$ are two $d$-dimensional real vectors. Let $\bfx^\downarrow$ be the vector obtained by arranging the components of $\bfx$ in nonincreasing order, which means $x^\downarrow_0\geq x^\downarrow_1\geq \cdots \geq x^\downarrow_{d-1}$; define $\bfy^\downarrow$ in a similar way.
Then $\bfx$ is \emph{majorized} by $\bfy$, denoted by $\bfx\prec \bfy$, if 
\begin{equation}
	\sum_{j=0}^{k} x^\downarrow_j \leq \sum_{j=0}^k y^\downarrow_j,\quad k=0,1, \ldots, d-1,  
\end{equation}
and the inequality is saturated when $k=d-1$. A function $f$ defined on a subset of $\bbR^d$ is \emph{Schur convex} (\emph{concave}) if $f(\bfx)\leq f(\bfy)$ [$f(\bfx)\geq f(\bfy)$] whenever $\bfx\prec \bfy$.

As a generalization,  given two pure states $|\Psi\>$ and $|\Upsilon\>$ in $\caH_{AB}$, $|\Psi\>$ is \emph{majorized} by $|\Upsilon\>$ if the Schmidt vector of $|\Psi\>$ is majorized by the Schmidt vector of $|\Upsilon\>$, that is, $\bfs_\Psi\prec \bfs_\Upsilon$. A function $f$ defined on pure states in $\caH_{AB}$ is \emph{Schur convex} (\emph{concave}) if it is invariant under local unitary transformations and is Schur convex (concave)
when regarded as a function of the Schmidt vector. For example, each entanglement measure $\caE_r$ defined in \eref{eq:Er} is Schur concave.

Now, we can formulate the majorization criterion on the transformations of bipartite pure states under LOCC originally established by Nielsen \cite{Nielsen99}.
\begin{proposition}\label{pro:Majorization}
Suppose $|\Psi\>$ and $|\Upsilon\>$ are two pure states in $\caH_{AB}$. Then $|\Psi\>$ can be transformed into $|\Upsilon\>$ under LOCC iff $|\Psi\>$ is majorized by $|\Upsilon\>$. 
\end{proposition}
\Pref{pro:Majorization} in particular implies that the Schmidt number and the entanglement measure $\caE_r$ cannot increase under LOCC.
If the majorization condition does not hold, then $|\Psi\>$ can be transformed into $|\Upsilon\>$ only probabilistically under LOCC. The maximum success probability $p(\Psi \rightarrow \Upsilon)$ 
is determined by Vidal~\cite{Vidal99}:
\begin{equation}
p(\Psi \rightarrow \Upsilon)=\min_r \frac{\caE_r(\Psi)}{\caE_r(\Upsilon)}.
\end{equation}
This result highlights the importance of the family of  entanglement measures $\caE_r$ in understanding bipartite entanglement transformations under LOCC.

\subsection{\label{ssec:QSV}Quantum state verification}

Before introducing the framework of  HDE certification, we briefly review the general framework of QSV \cite{PLM18,ZhuHEVPQSshort19,ZhuHEVPQSlong19}.
A device is supposed to produce the target quantum state $|\Psi\>\in\caH$, but actually produces $N$ states $\sigma_1, \sigma_2, \ldots, \sigma_N$ in $N$ runs.
We assume that either $\sigma_j=|\Psi\>\<\Psi|$ for all $j$ or $\<\Psi|\sigma_j|\Psi\>\leq 1-\varepsilon$ for all $j$.
Now our task is to distinguish the two cases with a given significance level $\delta$. 
To this end, for each $\sigma_j$ we can perform a binary  measurement $\left \{ \Pi_l,\bbone-\Pi_l \right \} $ chosen randomly with probability $p_l$ from a given set of measurements. Here the first outcome $\Pi_l$ corresponds to passing the test, while the second outcome corresponds to failing the test. To guarantee that the target state can always pass the test, the test operator $\Pi_l$ should satisfy the condition $\Pi_l|\Psi\>=|\Psi\>$.

The performance of the above verification strategy is characterized by the verification operator $\Omega=\sum_l p_l \Pi_l$. When $\<\Psi|\sigma|\Psi\>\leq 1-\varepsilon$, the maximum probability that $\sigma$ can pass one test  on average reads \cite{PLM18,ZhuHEVPQSshort19,ZhuHEVPQSlong19}
\begin{equation}
    \max_{\<\Psi|\sigma|\Psi\>\leq 1-\varepsilon} \tr(\Omega \sigma)=1-\varepsilon+\beta(\Omega)\varepsilon=1-\nu(\Omega)\varepsilon ,
\end{equation}
where $\beta(\Omega)$ is the second largest eigenvalue of $\Omega$, and $\nu(\Omega):=1-\beta(\Omega)$ is the \emph{spectral gap} from the maximum eigenvalue $1$. After $N$ runs, the probability that  ``bad'' states $\sigma_j$ pass all the tests is at most $[1-\nu(\Omega)\varepsilon]^N$.
To verify the target state $|\Psi\>$ within infidelity $\varepsilon$ and significance level $\delta$ using the strategy $\Omega$, the minimum number of tests required is \cite{PLM18,ZhuHEVPQSshort19,ZhuHEVPQSlong19}
\begin{equation}\label{eq:NQSV,ZhuHEVPQSlong19}
    N=\left\lceil\frac{\ln{\delta}}{\ln{[1-\nu(\Omega)\varepsilon}]}\right\rceil\approx \left\lceil\frac{1}{\nu(\Omega)} \varepsilon^{-1}\ln\left(\delta^{-1}\right)\right\rceil ,
\end{equation}
which is inversely proportional to the spectral gap $\nu(\Omega)$.

To minimize the number $N$ of tests, we need to maximize the spectral gap $\nu(\Omega)$. If there is no restriction on the  measurements, then the best strategy is to perform the test $\left \{|\Psi\>\<\Psi|,\bbone-|\Psi\>\<\Psi| \right \}$, which means $\Omega=|\Psi\>\<\Psi|$, $\nu(\Omega)=1$, and $N\approx\varepsilon^{-1}\ln\left(\delta^{-1}\right)$.
In most cases of practical interest, however, the target state $|\Psi\>$ is entangled, and it is extremely difficult to perform the entangling measurement mentioned above. Therefore, it is of paramount importance to determine the
efficiency limit of restricted operations, such as LOCC or separable operations as a relaxation. A verification operator $\Omega$ for $|\Psi\>$ is called local if the corresponding strategy can be realized  by LOCC. The verification operator $\Omega$ is separable if both $\Omega$ and $\bbone-\Omega$ are proportional to separable states. Define $\nu_\lc(\Psi)$ as the maximum spectral gap that can be achieved by LOCC and  $\beta_\lc(\Psi)=1-\nu_\lc(\Psi)$. Define $\nu_\sep(\Psi)$ and $\beta_\sep(\Psi)$ as the counterparts for separable operations.

Next, the verification operator $\Omega$ for $|\Psi\>$ is \emph{homogeneous} \cite{ZhuHEVPQSshort19,ZhuHEVPQSlong19} if it has the form
\begin{equation}
	\Omega=|\Psi\>\<\Psi|+\beta(\Omega)(\bbone-|\Psi\>\<\Psi|),
\end{equation}
where $\beta(\Omega)\in[0,1)$. In this case, the probability $\tr(\Omega\sigma)$ is completely determined by the fidelity $\<\Psi|\sigma|\Psi\>$. Such a verification strategy is also useful for fidelity estimation. In analogy to $\nu_\lc(\Psi)$,  we define $\nu_\lc^\rmH(\Psi)$ as the maximum spectral gap of homogeneous verification operators  that can be realized by LOCC and let $\beta_\lc^\rmH(\Psi)=1-\nu_\lc^\rmH(\Psi)$. As the counterparts for separable operations,  $\nu_\sep^\rmH(\Psi)$ and $\beta_\sep^\rmH(\Psi)$ can be defined in  a similar way. Here the argument $\Psi$ can be omitted if there is no danger of confusion.

\subsection{\label{ssec:VBPS}Verification of bipartite pure states}
In this subsection we summarize the main results on the verification of bipartite pure states. Our discussion is mainly based on \rscite{Owari08,ZhuHOVMES19,LiHZ19,Yu19}. Nevertheless, some results presented here were not clearly stated before.

Suppose  $|\Psi\>\in \caH_{AB}$ is a bipartite pure state, then 	the spectral gaps $\nu_\lc (\Psi)$, $\nu_\sep (\Psi)$, $\nu_\lc^\rmH (\Psi)$, and $\nu_\sep^\rmH (\Psi)$ are completely determined by the nonzero Schmidt coefficients of $|\Psi\>$ as we shall see shortly. In addition, they are closely tied to important entanglement measures  known as the robustness and random robustness \cite{VidaT99,Owari08,ZhuHOVMES19}. 
\Psref{pro:SpectralGapDimInd}-\ref{pro:QSVopt} below are proved in \aref{app:QSVProofs}.
\begin{proposition}\label{pro:SpectralGapDimInd}
	Suppose $|\Psi\>\in\caH_{AB}$; then the spectral gaps $\nu_\lc (\Psi)$, $\nu_\sep (\Psi)$, $\nu_\lc^\rmH (\Psi)$, and $\nu_\sep^\rmH (\Psi)$ only depend on the nonzero Schmidt coefficients of $|\Psi\>$ and are independent of the local dimensions. 
\end{proposition}

Thanks to \pref{pro:SpectralGapDimInd},  without loss of generality, we can assume that  $|\Psi\>$ has the following Schmidt decomposition:
\begin{equation}\label{eq:GBPS}
	|\Psi\>=\sum_{j=0}^{d-1} \sqrt{s_j}\lsp|jj\>.
\end{equation}
Then one can construct the following three verification operators for $|\Psi\>$, which are optimal or nearly optimal:
\begin{align}
	\Omega_\sep&:=|\Psi\>\<\Psi|+\sum_{j,k=0,\,j\neq k}^{d-1} \sqrt{s_js_k}\lsp |jk\>\<jk|,  \label{eq:OmegaSep} \\
	\Omega_\sep^\rmH&:=|\Psi\>\<\Psi|+\frac{\sqrt{s_0s_1}}{\lsp{1+\sqrt{s_0s_1}}}(\bbone-|\Psi\>\<\Psi|),  \label{eq:OmegaSepH} \\
	\Omega_\lc^\rmH&:=|\Psi\>\<\Psi|+\frac{s_0+ s_1}{2+s_0+ s_1}(\bbone-|\Psi\>\<\Psi|).  \label{eq:OmegaLCH}
\end{align}
Here, $\Omega_\sep$ was introduced in \rcite{Owari08} and is tied to the computation of the robustness of entanglement \cite{VidaT99}; $\Omega_\sep^\rmH$ is tied to the random robustness; $\Omega_\lc^\rmH$ was introduced in \rcite{LiHZ19}.
The next proposition follows from \rscite{Owari08,LiHZ19} as shown in \aref{app:QSVProofs}.  
\begin{proposition}\label{pro:OmegaLCsep}
	Suppose $|\Psi\>\in\caH_{AB}$ is given in \eref{eq:GBPS}; then $\Omega_\sep$ is  a separable verification operator of $|\Psi\>$,  $\Omega_\sep^\rmH$ is a separable homogeneous verification operator of $|\Psi\>$,  and $\Omega_\lc^\rmH$ is a local homogeneous verification operator of $|\Psi\>$. In addition,
	\begin{align}
		&\beta\left(\Omega_\sep\right)=\sqrt{s_0s_1}, &&\nu\left(\Omega_\sep\right)=1-\sqrt{s_0s_1}, \label{eq:GapSep}\\
		&\beta\left(\Omega_\sep^\rmH\right)=\frac{\sqrt{s_0s_1}}{1+\sqrt{s_0s_1}}, &&\nu\left(\Omega_\sep^\rmH\right)=\frac{1}{1+\sqrt{s_0s_1}}, \label{eq:GapSepH}\\
		&\beta\left(\Omega_\lc^\rmH\right)=\frac{s_0+ s_1}{2+s_0+ s_1}, &&\nu\left(\Omega_\lc^\rmH\right)=\frac{2}{2+s_0+ s_1}.  \label{eq:GapLCH} 
	\end{align}
\end{proposition}

\begin{proposition}\label{pro:betaOmegaLB}
	Suppose $|\Psi\>\in\caH_{AB}$ is given in \eref{eq:GBPS} and has Schmidt rank~$r$. Then 
	\begin{gather}
		\frac{\left(\sum_j \sqrt{s_j}\right)^2-1}{r^2-1}\leq \beta_\sep\leq \beta_\lc\leq \frac{s_0+s_1}{2+s_0+s_1}\leq\frac{1}{3}, \label{eq:betaOmegaLB1}\\
		\frac{\sqrt{s_0s_1}}{1+\sqrt{s_0s_1}}=\beta_\sep^\rmH\leq \beta_\lc^\rmH\leq \frac{s_0+s_1}{2+s_0+s_1}\leq\frac{1}{3}.  \label{eq:betaOmegaLB2}
	\end{gather}	
\end{proposition}
The inequalities in \eqsref{eq:betaOmegaLB1}{eq:betaOmegaLB2} are equivalent to the following inequalities:
\begin{gather}
\frac{2}{3}\leq \frac{2}{2+s_0+s_1}\leq \nu_\lc\leq \nu_\sep \leq 	\frac{r^2-\left(\sum_j \sqrt{s_j}\right)^2}{r^2-1},\label{eq:nuOmegaLB1}\\		
\frac{2}{3}\leq\frac{2}{2+s_0+s_1}\leq \nu_\lc^\rmH\leq \nu_\sep^\rmH=
\frac{1}{1+\sqrt{s_0s_1}}\label{eq:nuOmegaLB2}.
\end{gather}

According to \rcite{LiHZ19}, the verification strategy  $\Omega_\lc^\rmH$ can be realized using adaptive local projective measurements based on 2-designs \cite{Renes04,Scott06,Zauner11} and two-way communication. It is not clear if this strategy can be realized in a simpler way.
As an alternative, \rscite{ZhuHOVMES19,LiHZ19} proposed a simpler verification strategy using only two local projective tests based on MUB and one-way communication. Recall that two bases $\{|\psi_j\>\}_{j=0}^{d-1}$ and $\{|\varphi_j\>\}_{j=0}^{d-1}$ on $\caH_A$ are mutually unbiased if $|\<\psi_j|\varphi_k\>|^2=1/d$ for all $j,k=0,1,\ldots, d-1$. To be concrete, we assume that the MUB is composed of the computational basis $\{|j\>\}_{j=0}^{d-1}$ and the Fourier basis $\{|u_j\>\}_{j=0}^{d-1}$, where $|u_j\>=\sum_{k=0}^{d-1}\omega^{jk}|k\>/\sqrt{d}$ with  $\omega=\rme^{2\pi\rmi/d}$. 
Let  $M=\sqrt{d}\diag(\sqrt{s_0},\ldots,\sqrt{s_{d-1}})$ and $|v_j\>=M|u_j^*\>$. Then the verification operator tied to the MUB strategy is a convex sum of two test projectors and can be expressed as follows:
\begin{equation}\label{eq:OmegaMUB}
	\Omega_\MUB:=\frac{1}{2}\sum_{j=0}^{d-1} |jj\>\<jj|+\frac{1}{2}\sum_{j=0}^{d-1} |u_j v_j\>\<u_j v_j|.
\end{equation}
In addition, \rcite{LiHZ19} showed that
\begin{align}
\beta(\Omega_\MUB)=\frac{1}{2}, \quad  \nu(\Omega_\MUB)=\frac{1}{2}. \label{eq:Gap2}
\end{align}
Moreover, $\nu(\Omega_\MUB)$ attains the maximum among all strategies composed of two distinct tests based on local projective measurements. So the MUB strategy is quite appealing to practical applications.

Next, suppose  $|\Psi\>\in \caH_{AB}$ is maximally entangled, which means  $s_j=1/d$ for $j=0, 1,\ldots, d-1$. Then the first three inequalities in \eref{eq:betaOmegaLB1} and the first two inequalities in  \eref{eq:betaOmegaLB2} are saturated. In addition, the verification operators $\Omega_\sep^\rmH$ and $\Omega_\lc^\rmH$ coincide and are optimal among separable verification operators. 
Moreover, the optimal separable verification operator is unique when $d=d_A=d_B$ according to \pref{pro:QSVopt} below.

\begin{proposition}\label{pro:QSVopt}
	Suppose $|\Psi\>\in \caH_{AB}$ is a maximally entangled state.  Then 
\begin{align}\label{eq:betaOmegaOpt}
\beta_\sep= \beta_\lc=\beta_\sep^\rmH=\beta_\lc^\rmH=\frac{1}{d+1}.
\end{align}	
If in addition $d=d_A=d_B$ and $\Omega$ is a separable verification operator of $|\Psi\>$. Then $\beta(\Omega)=1/(d+1)$ iff
\begin{equation}\label{eq:PsiOmegaOpt}
	\Omega=\Omega_{\opt}:=|\Psi\>\<\Psi|+\frac{1}{d+1}(\bbone-|\Psi\>\<\Psi|),
\end{equation} 
which is automatically local and homogeneous. 
\end{proposition}

When $d=d_A=d_B$, $\caH_A$ and $\caH_B$ are isomorphic. 
According to \rcite{ZhuHOVMES19}, the optimal verification strategy $\Omega_{\opt}$ in \eref{eq:PsiOmegaOpt} can be realized using local projective measurements based on 2-designs.
Suppose $|\Psi\>=|\Phi\>$ and  $\{\caB_l,\lsp  p_l\}_l$ is a weighted set of orthonormal bases in $\caH_{A}$ that forms a complex projective 2-design. Then $\Omega_{\opt}$ can be realized as follows:
\begin{equation}
	\Omega_{\opt}=\sum_l p_l \Pi(\caB_l),\quad  \Pi(\caB_l)=\sum_{|\psi\>\in \caB_l}|\psi\>\<\psi|\otimes|\psi^*\>\<\psi^*|,
\end{equation}
where $|\psi^*\>$ denotes the complex conjugate of $|\psi\>$ with respect to the computational basis, and the test $\Pi(\caB_l)$ can be realized by conjugate-basis measurements. To construct such an optimal strategy, we need at least $d+1$ distinct local projective tests because at least $d+1$ bases are required to construct a 2-design in dimension $d$ \cite{Roy07}. The lower bound can be saturated when the local dimension $d$ is a prime power, in which case a 2-design can be constructed from a complete set of $d+1$ MUB (with uniform weights) \cite{Ivanovic81,William89,Durt10}.

\section{\label{sec:ResCert}Entanglement and resource certification}

In this section we introduce a general framework for certifying quantum resources under restricted operations. For concreteness, we shall focus on entanglement certification under LOCC or separable operations in the following discussions, but generalization to other resources, such as quantum coherence and nonstabilizerness, is immediate.

Let $\caH$ be a bipartite or multipartite Hilbert space  and $|\Psi\>\in \caH$ the target state of practical interest. Let $\caS$ be a closed convex subset of quantum states in $\caD(\caH )$ that represents the set of separable states or states with limited entanglement. Here our main concern is bipartite entanglement and we are particularly interested in the sets $\caS_r$ and $\ser(E)$ defined in \sref{ssec:HDE}, but the following analysis has a much wider scope of applications.

\subsection{\label{ssec:FW}Basic framework}
A device is supposed to produce the target  state $|\Psi\>\in \caH$, but may actually produce  states $\sigma_1,\sigma_2,\ldots, \sigma_N\in \caS$ in $N$ runs. Now, our task is to distinguish the two situations as efficiently as possible. To this end we can devise  a family of  binary tests $\{\Pi_l, \bbone-\Pi_l\}_l$ as in QSV, such that the target state can always pass each test, that is, $\Pi_l |\Psi\>=|\Psi\>$. Suppose the test $\{\Pi_l, \bbone-\Pi_l\}$ is chosen randomly with probability $p_l$; then the performance of this strategy is determined by the verification operator $\Omega=\sum_l p_l \Pi_l$ as in QSV. More precisely, the performance is determined by the \emph{separation probability} $P_\Omega(\Psi,\caS)$ defined as follows:
\begin{equation}\label{eq:Performance}
    P_\Omega(\Psi,\caS):=\max_{\sigma\in \caS} \tr(\Omega\sigma).
\end{equation}
The smaller the separation probability $P_\Omega(\Psi,\caS)$, the better the strategy $\Omega$ can distinguish the target state $|\Psi\>$ from states in the set $\caS$. After $N$ runs, the  states in the set $\caS$ can pass all tests with probability at most 
$P_\Omega^N(\Psi,\caS)$. When $P_\Omega(\Psi,\caS)<1$, to certify the target state $|\Psi\>$ with significance level $\delta$, which means $P_\Omega^N(\Psi,\caS)\leq \delta$,  it suffices to choose 
\begin{equation}\label{eq:NCertHDEOmega}
	N=\left\lceil \frac{\ln{\delta}}{\ln{P_\Omega(\Psi,\caS)}}\right \rceil =\left\lceil\frac{\ln\left(\delta^{-1}\right)}{\ln\left[P_\Omega^{-1}(\Psi,\caS)\right]}\right \rceil.
\end{equation}

To minimize the number of tests required, we need to  minimize $P_\Omega(\Psi,\caS)$ over all accessible strategies. Here we are particularly interested in strategies that can be realized by LOCC. The \emph{separation probability} of $|\Psi\>$ with respect to the set $\caS$  is defined as 
\begin{equation}\label{eq:SepProb}
    P_\lc(\Psi,\caS):=\min_{\Omega}P_\Omega(\Psi,\caS)=\min_{\Omega}\max_{\sigma\in \caS} \tr(\Omega\sigma),
\end{equation}
where the minimization is taken over all local strategies. It can be written as $P(\Psi,\caS)$ if there is no danger of confusion. If $\Omega$ is an optimal strategy, then the number of tests required reads
\begin{equation}\label{eq:NCertHDE}
	N=\left\lceil \frac{\ln{\delta}}{\ln{P(\Psi,\caS)}}\right \rceil =\left\lceil\frac{\ln\left(\delta^{-1}\right)}{\ln[P^{-1}(\Psi,\caS)]}\right \rceil.
\end{equation}
Therefore, the separation probability $P(\Psi,\caS)$
determines how well we can certify the target state $|\Psi\>$ against the set $\caS$ under LOCC and 
is thus of special interest to both theoretical studies and practical applications.
When $\caS=\caS_\sep$, the notation $P(\Psi,\caS)$ can be abbreviated as $P(\Psi)$ for simplicity. 

The separation probability $P_\lc^\rmH(\Psi,\caS)$ is defined in analogy to $P_\lc(\Psi,\caS)$ in \eref{eq:SepProb} except that the minimization is taken over all homogeneous strategies that can be realized by LOCC.  It can be abbreviated as $P^\rmH(\Psi,\caS)$ if there is no danger of confusion. 
In addition, $P_\sep(\Psi,\caS)$ and $P_\sep^\rmH(\Psi,\caS)$ can be defined by replacing LOCC with separable operations.

\subsection{\label{ssec:PropofSP}Basic properties of the separation probabilities}

Here we clarify the basic properties of the separation probabilities based on LOCC. Suppose $\caS$ is a closed convex subset of $\caD(\caH)$ that is invariant under local unitary transformations.  \Psref{pro:PassProbUB}-\ref{pro:OptOmegaSym} below can be verified by straightforward calculations. Similar results also apply to separation probabilities based on separable operations.

\begin{proposition}\label{pro:PassProbUB}
	Suppose $\Omega$ is  a verification operator of $|\Psi\>\in \caH$; then 
	\begin{equation}\label{eq:PassProbUB}
		P_\Omega(\Psi,\caS)\leq \nu(\Omega) F(\Psi,\caS)+\beta(\Omega),
	\end{equation}
	where the inequality is saturated when $\Omega$ is homogeneous. 
\end{proposition}
\Pref{pro:PassProbUB} is a simple corollary of the following inequality: 
\begin{equation}
\Omega\leq |\Psi\>\<\Psi|+\beta(\Omega)(\bbone-|\Psi\>\<\Psi|)=\nu(\Omega)|\Psi\>\<\Psi|+\beta(\Omega)\bbone,
\end{equation}
and the inequality is saturated when $\Omega$ is homogeneous. In turn it implies the following proposition. 
\begin{proposition}\label{pro:SepProbUBQSV}
Suppose $|\Psi\>\in \caH$; then 
\begin{gather}
F(\Psi,\caS)\leq P(\Psi,\caS)\leq \nu(\Psi) F(\Psi,\caS)+\beta(\Psi), \label{eq:SepProbLUB} \\
	P^\rmH(\Psi,\caS)=\nu^\rmH(\Psi) F(\Psi,\caS)+\beta^\rmH(\Psi). \label{eq:SepProbH}
\end{gather}
\end{proposition}

Next, we  show that pure states that are equivalent under local unitary transformations share the same separation probabilities. In addition, their optimal verification operators are connected by local unitary transformations. 
\begin{proposition}\label{pro:SepProbUT}
Suppose $\Omega$ is  a (homogeneous) verification operator of   $|\Psi\>\in \caH$ and $U\in \rmU(\caH)$. Then $U\Omega U^\dag$ is a  (homogeneous) verification operator of $U|\Psi\>$. If $U$ is a local unitary, then 
\begin{equation}
	\begin{aligned}
		P_{U\Omega U^\dag}\left(U\Psi U^\dag,\caS\right)&=P_{\Omega}(\Psi,\caS),\\
		P\left(U\Psi U^\dag,\caS\right)&=P(\Psi,\caS),\\
		P^\rmH\left(U\Psi U^\dag,\caS\right)&=P^\rmH(\Psi,\caS).
	\end{aligned}
\end{equation}
\end{proposition}
Thanks to \pref{pro:SepProbUT}, we can focus on a particular representative within each equivalent class when studying the separation probabilities. In addition, \pref{pro:SepProbUT} implies the following two propositions.

\begin{proposition}\label{pro:OmegaAve}
	Suppose $\Omega$ is a verification operator of the state $|\Psi\>\in \caH$, $U\in \rmU(\caH)$,  $U|\Psi\>=|\Psi\>$, $\Omega_1=U\Omega U^\dagger$, and  $\Omega_2=(\Omega+\Omega_1)/2$. Then $\Omega_1$ and $\Omega_2$  are  verification operators of $|\Psi\>$ and $\beta(\Omega_2)\leq \beta(\Omega_1)=\beta(\Omega)$. 
If $U$ is a local unitary, then 
\begin{equation}
	P_{\Omega_2}(\Psi,\caS)\leq P_{\Omega_1}(\Psi,\caS)=P_{\Omega}(\Psi,\caS).
\end{equation}
\end{proposition}

Let $\rmU_\rmL(\caH)$ be the group composed of all local unitaries in $\rmU(\caH)$ and  $\trmU_\rmL(\caH)$ the group generated by $\rmU_\rmL(\caH)$ and complex conjugation (with respect to the computational basis).
\begin{proposition}\label{pro:OptOmegaSym}
	Suppose $|\Psi\>\in \caH$ is invariant under a  subgroup $G$ of $\trmU_\rmL(\caH)$.  Then there exists an optimal local verification operator $\Omega$ of $|\Psi\>$ which is $G$-invariant, that is, 
	\begin{equation}
	P_\Omega(\Psi,\caS)=P(\Psi,\caS),\quad U\Omega U^\dag =\Omega\quad \forall \, U\in G. 
	\end{equation}
\end{proposition}
Thanks to \pref{pro:OptOmegaSym}, to determine the separation probability of $|\Psi\>$ or to construct an optimal verification operator, it suffices to consider verification operators that share the same symmetry as $|\Psi\>$. This simple observation is very helpful for studying the separation probabilities of bipartite and multipartite quantum states.

\section{\label{sec:CertHDE}Certification of high-dimensional entanglement}

In this section we consider the certification of HDE in general bipartite pure states in the bipartite Hilbert space $\caH_{AB}$ and clarify the basic properties of  separation probabilities.
The following proposition is  proved in \aref{app:SepProbProofs}.
\begin{proposition}\label{pro:SepProbDimInd}
	Suppose $|\Psi\>\in\caH_{AB}$, $\caS=\caS_r$ or $\caS=\ser(E)$ with $r\in [d-1]$, and $0\leq E< \caE_r(\Psi)$.
	Then the separation probabilities $P_\lc (\Psi,\caS)$, $P_\sep (\Psi,\caS)$, $P_\lc^\rmH (\Psi,\caS)$, and $P_\sep^\rmH (\Psi,\caS)$  are independent of the local dimensions as long as the nonzero Schmidt coefficients of $|\Psi\>$ are fixed. 
\end{proposition}
Thanks to \pref{pro:SepProbDimInd}, to determine the separation probabilities of $|\Psi\>$ we can consider any bipartite pure state that has the same nonzero Schmidt coefficients as $|\Psi\>$, which is very helpful for simplifying the analysis.

\subsection{\label{ssec:CertHDEinMES}HDE in maximally entangled states}

To start with, here we focus on the maximally entangled state $|\Phi\>\in \caH_{AB}$ defined in \eref{eq:MES}, assuming that $d=d_A=d_B$.  Any other maximally entangled state in $\caH_{AB}$ is equivalent to $|\Phi\>$ under local unitary transformations. Note that $|\Phi\>$ is invariant under arbitrary  unitary transformations of the form $U\otimes U^*$ with $U\in \rmU(d)$. Moreover, any operator on $\caH_{AB}$ that shares this symmetry is a linear combination of $|\Phi\>\<\Phi|$ and the identity operator. Therefore, by \pref{pro:OptOmegaSym}, to construct an optimal verification operator, it suffices to consider homogeneous verification operators. Moreover, the verification operator $\Omega_\opt$ defined in \pref{pro:QSVopt} with $|\Psi\>=|\Phi\>$ is optimal for HDE certification among all verification operators based on  separable operations, including LOCC.

Now, suppose $\caS$ is a closed convex subset  of $\caD(\caH_{AB})$ that is invariant under local unitary transformations. 
Then, by \pref{pro:SepProbUBQSV} with $\beta(\Phi)=\beta(\Omega_\opt)=1/(d+1)$  and $\nu(\Phi)=\nu(\Omega_\opt)=d/(d+1)$, the separation probability $P(\Phi,\caS)$ reads
\begin{equation}\label{eq:SepProbMES}
	P(\Phi,\caS)=P_{\Omega_\opt}(\Phi,\caS)=\frac{d}{d+1}F(\Phi, \caS)+\frac{1}{d+1}.
\end{equation}
The following theorem is a simple corollary of \eref{eq:SepProbMES} and \pref{pro:FPsiSk}.
\begin{theorem}\label{thm:SepProbforEMCinMES}
Suppose $r\in [d-1]$ and $0\leq E< (d-r)/d$. Then 
\begin{align}
P(\Phi,\caS_r)&=\frac{r+1}{d+1}, \label{eq:SepProbforSNCinMES}\\
P(\Phi,\caS_{\caE_r}(E))&=\frac{\left[\sqrt{(d-r)E}+\sqrt{r(1-E)}\lsp\right]^2+1 }{d+1}. \label{eq:SepProbforEMCinMES}
\end{align}
\end{theorem}

\begin{figure}[t]
    \centering
    \includegraphics[scale=0.316]{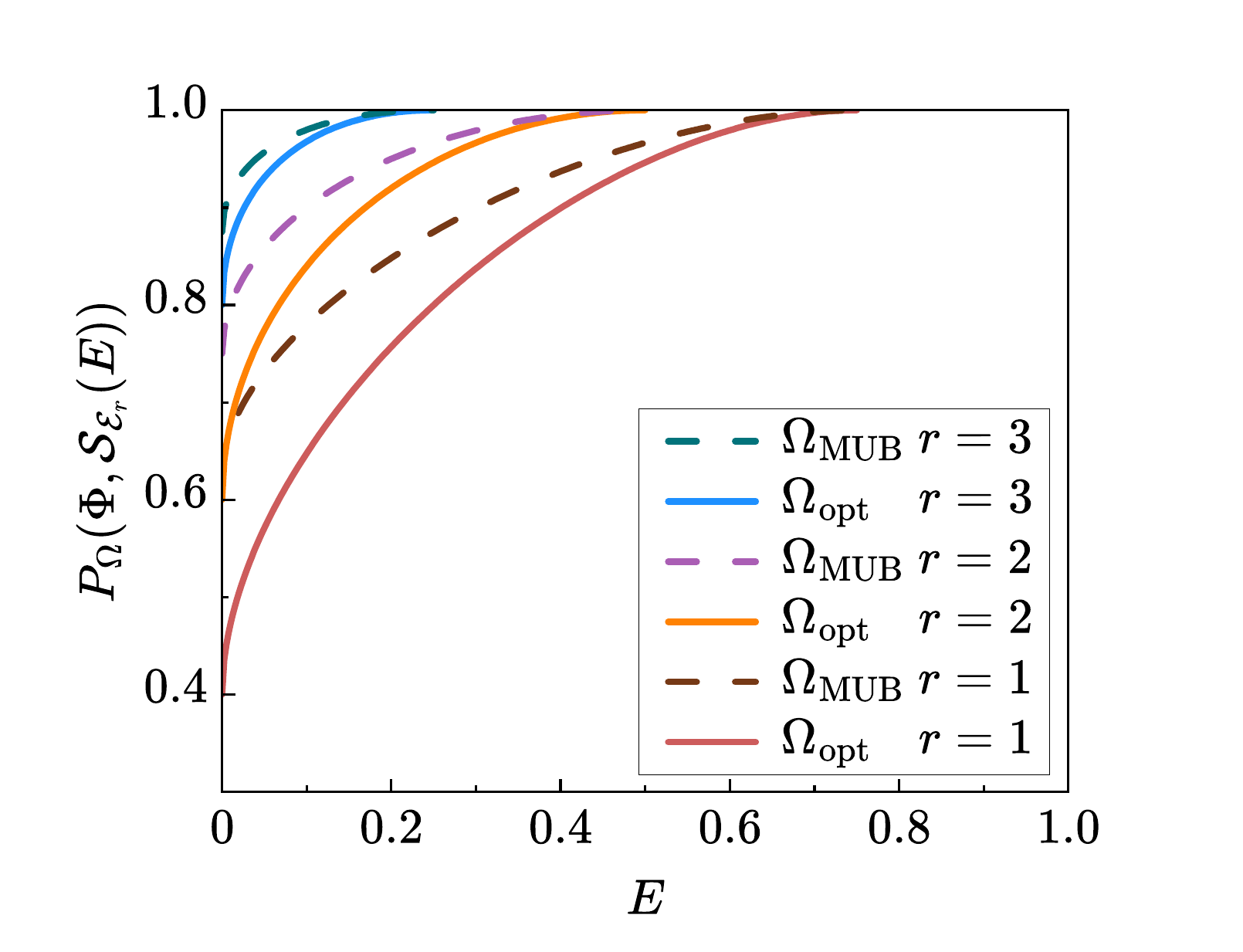}
        \caption{The separation probabilities  $P_\Omega(\Phi,\caS_{\caE_r}(E))$ achieved by the MUB strategy $\Omega_\MUB$ and optimal strategy $\Omega_\opt$ as functions of $E$ for $d=4$ and $r=1,2,3$.        	
}
    \label{fig:Entmon}
\end{figure}

In conjunction with \eref{eq:NCertHDE}, it is straightforward to determine the number of tests required to certify the lower bound $r+1$ for the Schmidt number with significance level $\delta$:
\begin{equation}\label{eq:NSNC}
    N=\left\lceil\left(\ln{\frac{d+1}{r+1}}\right)^{-1}\ln\left(\delta^{-1}\right)\right\rceil.
\end{equation}
This number increases monotonically with $r$, but decreases monotonically with $d$. When $d+1\geq (r+1)/\delta$, we have $N=1$, hence HDE can be certified using a single test when $d$ is sufficiently large compared with $r$. This conclusion is robust to noise thanks to the formula for  $P(\Phi,\caS_{\caE_r}(E))$ in \eref{eq:SepProbforEMCinMES}, as illustrated in \fref{fig:Entmon}. Note that the separation probability $P(\Phi,\caS_{\caE_r}(E))$ is nondecreasing in $r$ and $E$ and concave in $E$. In the special case $r=1$ and $E=0$, which means $\ser(E)=\caS_r=\caS_\sep$,  we obtain the separation probability for certifying entanglement and the corresponding number of required tests:
\begin{equation}\label{eq:NEC}
P(\Phi)=\frac{2}{d+1},\quad 	N=\left\lceil\left(\ln{\frac{d+1}{2}}\right)^{-1}\ln\left(\delta^{-1}\right)\right\rceil.
\end{equation}

To realize the optimal strategy $\Omega_\opt$ discussed above, we need to construct at least $d+1$ distinct local projective tests, which might be quite challenging. For comparison, next, we turn to the  simpler strategy $\Omega_\MUB$ originally introduced for QSV in \rscite{ZhuHOVMES19,LiHZ19} as mentioned in \sref{ssec:VBPS}, which is based on only two distinct local projective tests. Previously, similar strategies have been explored for certifying HDE \cite{Huang16,Erker17,Bavaresco18,Valencia20}, but the distinction between $\Omega_\MUB$ and $\Omega_\opt$ is not clear. The following proposition clarifies the separation probabilities achieved by $\Omega_\MUB$; it may be regarded as a special case of \pref{pro:Omega2} presented in \sref{ssec:CertHDEinGBPS}. 

\begin{proposition}\label{pro:Omega2MUB}
	Suppose $r\in[d-1]$ and $0\leq E< (d-r)/d$. Then
	\begin{align}
	P_{\Omega_\MUB}(\Phi,\caS_r)&=\frac{r+d}{2d}, \\
	P_{\Omega_\MUB}(\Phi,\caS_{\caE_r}(E))&=\frac{\left[\sqrt{(d-r)E}+\sqrt{r(1-E)}\lsp\right]^2+d}{2d}.
	\end{align}
\end{proposition}

In conjunction with \eref{eq:NCertHDE}, it is straightforward to determine the number of tests required by the strategy $\Omega_\MUB$ to certify the Schmidt number of the maximally entangled state $|\Phi\>$. \Fref{fig:2} illustrates the performance of $\Omega_\MUB$ in comparison with $\Omega_\opt$.  Although $\Omega_\MUB$ is not as efficient as $\Omega_\opt$, still it enables us to certify HDE using a constant sample cost.

\begin{figure}[t]
	\centering
	\includegraphics[scale=0.317]{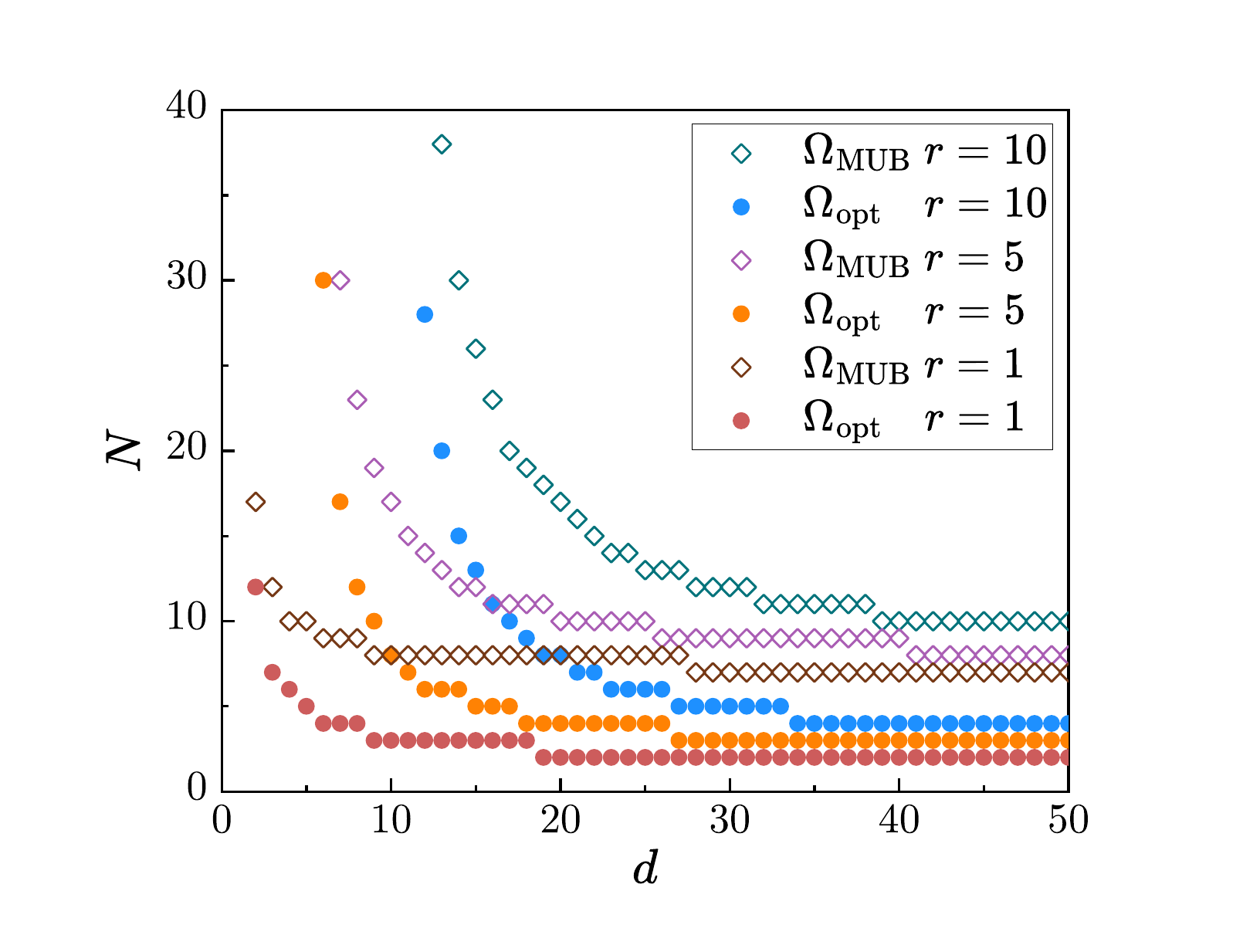}
	\caption{Certification of the Schmidt number of the $d\times d$ maximally entangled state based on the two strategies $\Omega_\MUB$ and $\Omega_\opt$. The figure shows the number of tests required by each strategy to certify the lower bound $r+1$ for the Schmidt number with significance level $\delta=0.01$. }
	\label{fig:2}
\end{figure}

\subsection{\label{ssec:CertHDEinGBPS}HDE in general bipartite pure states}

Now, we turn to a general bipartite pure state $|\Psi\>$ in $\caH_{AB}$, without assuming $d_A=d_B$. Since $|\Psi\>$ lacks the high degree of symmetry enjoyed by the maximally entangled state $|\Phi\>$, it is in general extremely difficult to derive an exact formula for the separation probabilities $P_\lc(\Psi,\caS)$ and $P_\sep(\Psi,\caS)$ with $\caS=\caS_r$ or $\caS=\ser(E)$, assuming that $r\in [d-1]$ and $0\leq E< \caE_r(\Psi)$. Nevertheless, we can clarify their key properties and derive nearly tight upper and lower bounds. 
\Thsref{thm:SepProb&Major}-\ref{thm:ConcSepProbforSNCinGBPS} below are proved in \aref{app:SepProbGBPSproofs}.

\begin{theorem}\label{thm:SepProb&Major}
	Suppose $|\Psi\>\in \caH_{AB}$, $r\in [d-1]$, $0\leq E< (d-r)/d$, and  $\caS=\caS_r$ or $\caS=\ser(E)$.
	Then the separation probabilities $P_\lc(\Psi,\caS)$ and $P_\sep(\Psi,\caS)$ are Schur convex in $|\Psi\>$ and do not decrease if $|\Psi\>$ is subjected to LOCC. In addition,  
	\begin{align}
	\!\!P_\lc(\Phi,\caS)=P_\sep(\Phi,\caS)\leq P_\sep(\Psi,\caS)\leq P_\lc(\Psi,\caS), \label{eq:SepProbPsiPhi}\\
	\!\!P_\lc(\Phi,\caS)=P_\sep(\Phi,\caS)\leq P_\sep^\rmH(\Psi,\caS)\leq P_\lc^\rmH(\Psi,\caS). \label{eq:SepProbPsiPhiH}
	\end{align}		
\end{theorem}

According to \thref{thm:SepProb&Major}, if $|\Psi\>$ is subjected to LOCC, then it is more difficult to distinguish the resulting state from  states with no or limited entanglement as expected. However, the situation is quite different if we restrict the analysis to homogeneous verification strategies. 
Notably, the separation probability $P_\sep^\rmH(\Psi,\caS_r)$ is not necessarily Schur convex. For example, consider the following two two-qutrit states:
\begin{equation}
\begin{aligned}
|\Psi_1\>&=\sqrt{\frac{2}{5}}\lsp|00\>+\sqrt{\frac{2}{5}}\lsp|11\>+\sqrt{\frac{1}{5}}\lsp|22\>,\\
|\Psi_2\>&=\sqrt{\frac{3}{5}}\lsp|00\>+\sqrt{\frac{1}{5}}\lsp|11\>+\sqrt{\frac{1}{5}}\lsp|22\>, 
\end{aligned}
\end{equation} 
where $|\Psi_1\>$ is majorized by $|\Psi_2\>$, which means  $|\Psi_1\>$ can be transformed into  $|\Psi_2\>$ under LOCC by \pref{pro:Majorization}. On the other hand, when $r=2$, from \thref{thm:SepProbGBPSLUB} below, we can deduce that
\begin{align} P_\sep^\rmH(\Psi_2,\caS_r)=\frac{4+\sqrt{3}}{5+\sqrt{3}}<P_\sep^\rmH(\Psi_1,\caS_r)=\frac{6}{7}. 
\end{align}
Therefore, $P_\sep^\rmH(\Psi,\caS_r)$ is not Schur convex.

\begin{figure*}[t]
	\centering
	\includegraphics[scale=0.41]{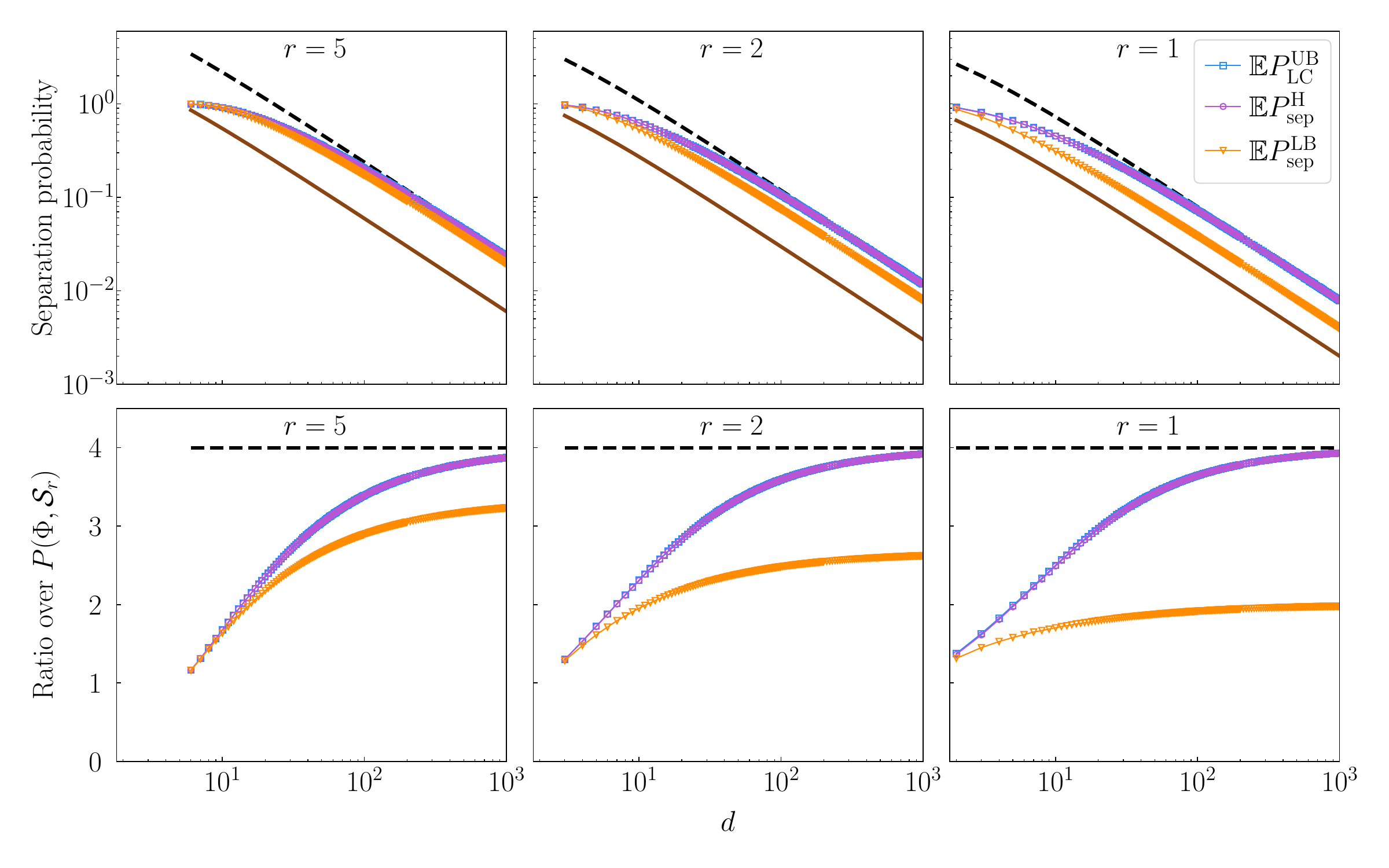}
	\caption{The mean separation probabilities (upper plots)
and their ratios over $P(\Phi,\caS_r)=P_\sep(\Phi,\caS_r)=(r+1)/(d+1)$ (lower plots)		
as functions of $d=d_A=d_B$ for $r=1,2,5$. Here $\bbE P_\lc^\ub(\Psi,\caS_r)$ is an upper bound for $\bbE P_\lc^\rmH(\Psi,\caS_r)$ and $\bbE P_\sep^\rmH(\Psi,\caS_r)$, while
 $\bbE P_\sep^\lb(\Psi,\caS_r)$ is a lower bound for $\bbE P_\sep(\Psi,\caS_r)$ and $\bbE P_\sep^\rmH(\Psi,\caS_r)$. Note that $\bbE P_\lc^\ub(\Psi,\caS_r)$ and $\bbE P_\sep^\rmH(\Psi,\caS_r)$ almost coincide.
  The brown solid line and black dashed line in each plot denote the lower bound $P(\Phi,\caS_r)$ and  upper bound $4P(\Phi,\caS_r)$, respectively, as presented in \thref{thm:SepProbGBPSmeanLUB}.
For each local dimension $d$,  $10000$ Haar-random pure states are sampled. }
	\label{fig:ConGBPS}
\end{figure*}

\begin{figure*}[t]
	\centering
	\includegraphics[scale=0.41]{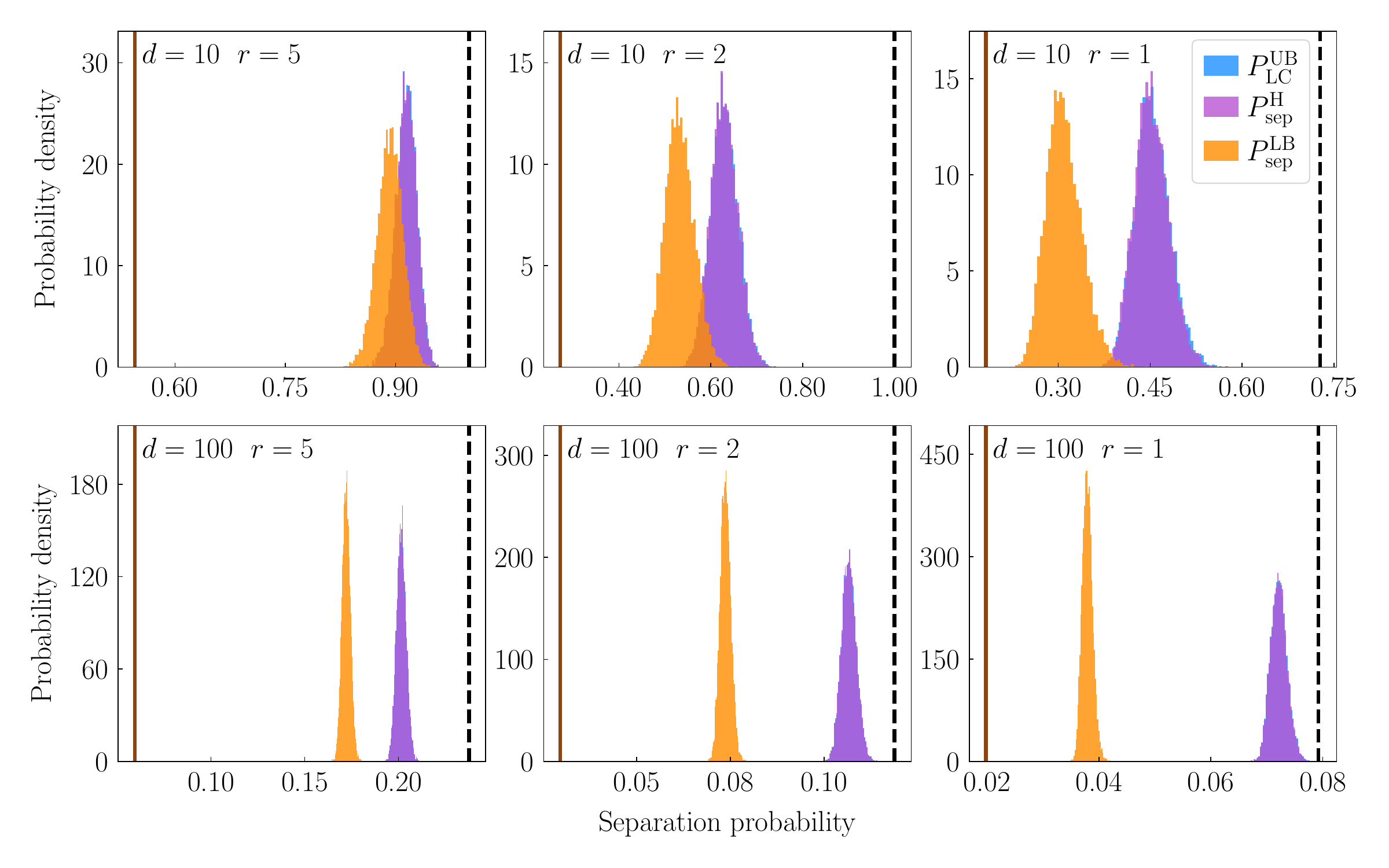}
	\caption{Distributions of the separation probabilities and relevant bounds for Haar-random pure states with $d=d_A=d_B=10,100$ and $r=1,2,5$. Note that the distributions of  $P_\lc^\ub(\Psi,\caS_r)$ and $P_\sep^\rmH(\Psi,\caS_r)$ are almost indistinguishable.
 The brown solid line and black dashed  line in each plot denote the values $P(\Phi,\caS_r)$ and $4P(\Phi,\caS_r)$, respectively (cf. \fref{fig:ConGBPS}). For each local dimension $d$,  $10000$ Haar-random pure states are sampled.}
	\label{fig:ProbGBPS}
\end{figure*}

Next, suppose  $|\Psi\>$ has Schmidt spectrum $\{s_j\}_{j=0}^{d-1}$ with $s_0\geq s_1 \geq \cdots \geq s_{d-1} \geq 0$.
Here, we provide a number of informative upper and lower bounds for four types of separation probabilities and show that HDE in a general bipartite pure state can be certified efficiently. Define
\begin{equation}\label{eq:PsepLBUB}
\begin{aligned}
P_\sep^\lb(\Psi,\caS_r)&=1-\caE_r(\Psi),\\
P_\lc^\ub(\Psi,\caS_r)&=1-\frac{2\caE_r(\Psi)}{2+s_0+ s_1}, \\
P_\sep^\lb(\Psi,\ser(E))&=f_r(\Psi,E), \\
P_\lc^\ub(\Psi,\ser(E))&=\frac{2f_r(\Psi,E)+s_0+s_1}{2+s_0+ s_1}, 
\end{aligned}
\end{equation}
where $\caE_r(\Psi)=\sum_{j=r}^{d-1}s_j$ and $f_r(\Psi,E)=F(\Psi,\ser(E))$ as determined in \eref{eq:FPsiSkE} of \pref{pro:FPsiSk}.
The following theorem is a direct consequence of \psref{pro:FPsiSk}, \ref{pro:betaOmegaLB}, \ref{pro:PassProbUB}, and~\ref{pro:SepProbUBQSV} as shown in \aref{app:SepProbGBPSproofs}.

\begin{widetext}
\begin{theorem}\label{thm:SepProbGBPSLUB}
	Suppose $|\Psi\>\in \caH_{AB}$ has Schmidt spectrum $\{s_j\}_{j=0}^{d-1}$, $r\in [d-1]$, and $0\leq E< \caE_r(\Psi)$. Then   
	\begin{gather}
		P_\sep^\lb(\Psi,\caS_r)\leq P_\sep(\Psi,\caS_r)\leq P_\lc(\Psi,\caS_r) \leq P_\lc^\ub(\Psi,\caS_r)\leq \frac{2}{1+s_0}P_\sep^\lb(\Psi,\caS_r), \label{eq:SepProbGBPSLUB}\\
		1-\frac{\caE_r(\Psi)}{1+\sqrt{s_0 s_1}}=P_\sep^\rmH(\Psi,\caS_r)\leq P_\lc^\rmH(\Psi,\caS_r)\leq P_\lc^\ub(\Psi,\caS_r)\leq \frac{3}{2+s_0}P_\sep^\rmH(\Psi,\caS_r), \label{eq:SepProbGBPSLUBH}\\
		P_\sep^\lb(\Psi,\ser(E))\leq P_\sep(\Psi,\ser(E))\leq P_\lc(\Psi,\ser(E)) \leq P_\lc^\ub(\Psi,\ser(E))\leq \frac{2}{1+s_0}P_\sep^\lb(\Psi,\ser(E)), \label{eq:SepProbErGBPSLUB}\\
		\frac{f_r(\Psi,E)+\sqrt{s_0 s_1}}{1+\sqrt{s_0 s_1}}=P_\sep^\rmH(\Psi,\ser(E))\leq P_\lc^\rmH(\Psi,\ser(E))\leq P_\lc^\ub(\Psi,\ser(E))\leq \frac{3}{2+s_0}P_\sep^\rmH(\Psi,\ser(E)).  \label{eq:SepProbErGBPSLUBH}
	\end{gather}
\end{theorem}
\end{widetext}

If $|\Psi\>=\sum_{j=0}^{d-1} \sqrt{s_j}\lsp|jj\>$ as in \eref{eq:GBPS}, then $P_\lc^\ub(\Psi,\caS_r)$ and $P_\lc^\ub(\Psi,\ser(E))$ correspond to the separation probabilities of the local homogeneous verification operator $\Omega_\lc^\rmH$ defined in \eref{eq:OmegaLCH}, which is reproduced from \rcite{LiHZ19}.
By contrast, $P_\sep^\lb(\Psi,\caS_r)$ and $P_\sep^\lb(\Psi,\ser(E))$ correspond to the separation probabilities of the verification operator $\Omega=|\Psi\>\<\Psi|$, which is in general not separable.
Nevertheless, for some bipartite entangled pure states, these separation probabilities can also be attained by a separable verification operator, as we shall see in \sref{sec:CertHDEin2qPS}.
The separation probabilities $P_\sep^\rmH(\Psi,\caS_r)$ and $P_\sep^\rmH(\Psi,\ser(E))$ can be attained by the optimal separable homogeneous verification operator $\Omega_\sep^\rmH$ in \eref{eq:OmegaSepH} (see also \pref{pro:OmegaLCsep}).

Next, we clarify the separation probabilities of the strategy $\Omega_\MUB$ defined in \eref{eq:OmegaMUB} for certifying the Schmidt number of the state $|\Psi\>$ in \eref{eq:GBPS}. 
Let
\begin{equation}
	|\tPsi\>=\sum_{j=0}^{r-1} \sqrt{\frac{(1-E)s_j}{1-\caE_r(\Psi)}}\lsp|jj\>+\sum_{j=r}^{d-1} \sqrt{\frac{E s_j}{\caE_r(\Psi)}}\lsp|jj\>;
\end{equation}
then $\tPsi\in \ser(E)$ and 
\begin{align}
	\tr\left(\Omega_\MUB |\tPsi\>\<\tPsi|\right)&=\frac{f_r(\Psi,E)+1}{2}.
\end{align}
Therefore, the inequality in \eref{eq:PassProbUB} of \pref{pro:PassProbUB} is saturated when $\Omega=\Omega_\MUB$ and 
$\caS=\ser(E)$, given that $\nu(\Omega_\MUB)=\beta(\Omega_\MUB)=1/2$. In conjunction with the observation $\caS_r=\ser(0)$, we can immediately deduce the following proposition.

\begin{proposition}\label{pro:Omega2}
	Suppose $|\Psi\>\in \caH_{AB}$ has Schmidt spectrum $\{s_j\}_{j=0}^{d-1}$, $r\in [d-1]$, and $0\leq E< \caE_r(\Psi)$. Then
	\begin{align}
	P_{\Omega_\MUB}(\Psi,\caS_r)&=1-\frac{\caE_r(\Psi)}{2}, \\
	P_{\Omega_\MUB}(\Psi,\caS_{\caE_r}(E))&=\frac{f_r(\Psi,E)+1}{2}.
	\end{align}
\end{proposition}

In conjunction with \eref{eq:PsepLBUB} we can derive the following relations: 
\begin{align}
	P_{\Omega_\MUB}(\Psi,\caS_r)&=\frac{P_\sep^\lb(\Psi,\caS_r)+1}{2}, \\
	P_{\Omega_\MUB}(\Psi,\caS_{\caE_r}(E))&=\frac{P_\sep^\lb(\Psi,\ser(E))+1}{2}.
\end{align}
These equations mean $P_{\Omega_\MUB}(\Psi,\caS_r)-P_{\sep}(\Psi,\caS_r)<1/2$ and the inequality still holds if $\caS_r$ is replaced by $\ser(E)$.

Next, we consider the general behavior of the separation probabilities for a Haar-random bipartite pure state  $|\Psi\>\in \caH_{AB}$.
With high probability we have $s_0,s_1 \sim 1/d$, $\caE_r(\Psi)\sim (d-r)/d$, and $P_\sep^\lb(\Psi,\caS_r)$, $P_\lc^\ub(\Psi,\caS_r)\sim r/d$ for $r<(d-3)/4$.
Therefore, HDE can be certified efficiently by LOCC, which is corroborated by rigorous analytical derivation and extensive numerical simulations as shown below.

\begin{theorem}\label{thm:SepProbGBPSmeanLUB}
	Suppose $|\Psi\>$ is a Haar-random pure state in $\caH_{AB}$ and $r\in [d-1]$. Then   
	\begin{align}
		\frac{r+1}{d+1}&\leq \bbE P_\sep(\Psi,\caS_r)\leq \bbE P_\lc(\Psi,\caS_r) \leq \bbE P_\lc^\rmH(\Psi,\caS_r) \nonumber\\
		&\leq \bbE P_\lc^\ub(\Psi,\caS_r)\leq \frac{4(r+1)}{d+1}=4P(\Phi,\caS_r). \label{eq:SepProbGBPSmeanLUB}
	\end{align}
\end{theorem}

\begin{theorem}\label{thm:ConcSepProbforSNCinGBPS}
	Suppose $|\Psi\>$ is a Haar-random pure state in $\caH_{AB}$, $r\in [d-1]$, and $\epsilon\geq0$. Then the separation probability $ P_\lc^\rmH(\Psi,\caS_r)$ satisfies
	\begin{equation}\label{eq:ConcSepProbforSNCinGBPS}
	\!\!\Pr\left\{ P_\lc^\rmH(\Psi,\caS_r)\geq \frac{4(r+1)}{d+1}+\epsilon \right\} \leq 2\exp\left( -\frac{D\epsilon^2}{50\pi}\right),
	\end{equation}
	where $D=\dim(\caH_{AB})$, and
	the same  result still holds if $P_\lc^\rmH(\Psi,\caS_r)$ is replaced by $P_\lc(\Psi,\caS_r)$, $P_\sep^\rmH(\Psi,\caS_r)$, or $P_\sep(\Psi,\caS_r)$. 
\end{theorem}

As a complement to  \thsref{thm:SepProbGBPSmeanLUB} and \ref{thm:ConcSepProbforSNCinGBPS}, 
\fref{fig:ConGBPS} shows the mean separation probability $\bbE P_\sep^\rmH(\Psi,\caS_r)$ as a function of $d$ for $r=1,2,5$.
Also shown are the upper bound $\bbE P_\lc^\ub(\Psi,\caS_r)$ for $\bbE P_\lc^\rmH(\Psi,\caS_r)$ and $\bbE P_\sep^\rmH(\Psi,\caS_r)$ and the lower bound $\bbE P_\sep^\lb(\Psi,\caS_r)$ for $\bbE P_\sep(\Psi,\caS_r)$ and $\bbE P_\sep^\rmH(\Psi,\caS_r)$. This figure indicates that
\begin{align}
\bbE P_\lc^\ub(\Psi,\caS_r)\approx \bbE P_\lc^\rmH(\Psi,\caS_r)\approx\bbE P_\sep^\rmH(\Psi,\caS_r),
\end{align}
so the  verification strategy in  \eref{eq:OmegaLCH} is nearly optimal among local homogeneous strategies for the Haar-random pure state $|\Psi\>$. In addition, the values in the above equation are close to the upper bound $4P(\Phi,\caS_r)$ when $d$ is sufficiently large compared with $r$, in which case the last two inequalities in \eref{eq:SepProbGBPSmeanLUB} are approximately saturated. \Fref{fig:ProbGBPS}  shows  the probability density distributions of $P_\lc^\ub(\Psi,\caS_r)$, $P_\sep^\rmH(\Psi,\caS_r)$, and $P_\sep^\lb(\Psi,\caS_r)$. All these distributions are concentrated within the interval $[P(\Phi,\caS_r),4P(\Phi,\caS_r)]$ and become increasingly concentrated as $d$ increases.  These results further demonstrate that HDE in general bipartite pure states can be certified efficiently.

\section{\label{sec:CertHDEin2qPS}Certification of entanglement in two-qubit pure states}

In this section, we construct an optimal separable strategy and two nearly optimal local strategies for certifying the entanglement of a general two-qubit pure state. When the state has sufficiently high concurrence, we further show that the optimal separable strategy can also be implemented by LOCC.

Up to a local unitary transformation, any two-qubit entangled pure state can be expressed as follows:
\begin{equation}\label{eq:2qPS}
	|\Psi_\theta\>=\cos\theta|00\>+\sin\theta|11\>,\quad 0< \theta\leq \frac{\pi}{4}. 
\end{equation}
Note that $|\Psi_\theta\>$ is invariant under the swap operation and complex conjugation (with respect to the computational basis). In addition, $|\Psi_\theta\>$ is invariant under any local unitary transformation of the form $V_\zeta\otimes V_\zeta^*$, where 
\begin{align}\label{eq:Vzeta}
	V_\zeta=|0\>\<0|+\rme^{-\rmi \zeta}|1\>\<1|, \quad 0\leq \zeta< 2\pi. 
\end{align}
Therefore, we can restrict our attention to verification operators that enjoy the same symmetry when searching for an optimal  strategy.

Given $0< \theta\leq \pi/4$ and $0\leq p\leq 1$, define the following verification operators:
\begin{align}
	&\Omega_0:=|\Psi_\theta\>\<\Psi_\theta|+\cos\theta\sin\theta(|01\>\<01|+|10\>\<10|), \label{eq:Omega0}\\
	&\Omega_1:=|\Psi_\theta\>\<\Psi_\theta|+\frac{\cos\theta\sin\theta}{1+\cos\theta\sin\theta}(\bbone-|\Psi_\theta\>\<\Psi_\theta|), \label{eq:Omega1} \\
	&\Omega(\theta,p):=p\Omega_1+(1-p)\Omega_0. \label{eq:Omegathetap}
\end{align}
Note that $\Omega_0=\Omega(\theta,0)$ and $\Omega_1=\Omega(\theta,1)$. In addition, $\Omega_0$, $\Omega_1$, and $\Omega(\theta,p)$ are invariant under the swap operation, complex conjugation, and any local unitary transformation of the form $V_\zeta\otimes V_\zeta^*$. 
Furthermore, by virtue of the PPT criterion \cite{Peres96,Horo96}, it is straightforward to verify that both $\Omega_0$ and $\Omega_1$ are separable. 
Here $\Omega_1$ is the optimal QSV strategy proposed in \rcite{WangH19}, which can be implemented by local operations with two-way classical communication, and 
$\Omega_0$ is a separable strategy proposed in the proofs of Lemma~1 and Theorem~2 of \rcite{Owari08}; however, no LOCC implementation of $\Omega_0$ has been found.
To simplify the following discussion, the separation probability $P_{\Omega(\theta,p)}(\Psi_\theta)$ will be abbreviated as $P(\theta,p)$ henceforth.

By virtue of \thref{thm:SepProbGBPSLUB} and the above observation, we can immediately derive the following proposition, which clarifies the separation probability associated with an optimal homogeneous strategy that can be realized by separable operations or LOCC. The result is illustrated in \fref{fig:2q}.
\begin{proposition}
	Suppose $0< \theta\leq \pi/4$. Then the local strategy $\Omega_1$ defined in \eref{eq:Omega1} is  optimal among all separable homogeneous strategies for $|\Psi_\theta\>$. The separation probabilities $	P_\sep^\rmH(\Psi_\theta)$ and $	P_\lc^\rmH(\Psi_\theta)$ read
	\begin{align}
		P_\sep^\rmH(\Psi_\theta)=P_\lc^\rmH(\Psi_\theta)=P_{\Omega_1}(\Psi_\theta)=
		\frac{\cos^2\theta+\cos\theta\sin\theta}{1+\cos\theta\sin\theta}.
	\end{align}
	If  $0\leq E<\sin^2\theta$, then 
	\begin{align}
		P_\sep^\rmH(\Psi_\theta,\serone(E))\!&=\!P_\lc^\rmH(\Psi_\theta,\serone(E))\!=\!P_{\Omega_1}(\Psi_\theta,\serone(E))\nonumber\\
		\!&=\!\frac{\cos^2(\theta-\theta_E)+\cos\theta\sin\theta}{1+\cos\theta\sin\theta},
	\end{align}
	where $\theta_E=\arcsin{\sqrt{E}}$.
\end{proposition}

Next, we turn to general verification strategies. For $0<\theta\leq \pi/4$ and $0\leq p\leq 1$, define
\begin{equation}\label{eq:thetaFun}
	\begin{aligned}
		\kappa(\theta)&:=\cos\theta\sin\theta,\\
		q(\theta)&:=\max\left\{\frac{(1+\kappa)(2\kappa-\cos^2\theta)}{\kappa(2\kappa+\sin^2\theta)},0\right\},\\
		a^*(\theta,p)&:=\begin{cases}
			0 &	q(\theta)\leq p\leq 1,\\
			\arctan{\sqrt{h(\theta,p)}} & 
			0\leq p <q(\theta),
		\end{cases} \\
		h(\theta,p)&:=\frac{(1+\kappa)(2\kappa-\cos^2\theta)-p\kappa(2\kappa+\sin^2\theta)}{(1+\kappa)(2\kappa-\sin^2\theta)-p\kappa(2\kappa+\cos^2\theta)},	
	\end{aligned}
\end{equation}
where the arguments $\theta$ and $p$ can be omitted to simplify the notation. Here $q(\theta)=0$ when $0<\theta\leq \arctan(1/2)$ and $q(\theta)$ increases  from 0 to 1 when 
$\theta$ increases from $\arctan(1/2)$ to $\pi/4$.
In addition, if $\arctan(1/2)< \theta\leq \pi/4$ and $0\leq p <q(\theta)$, then we have $0< h(\theta,p)\leq 1$ and thus $0< a^*(\theta,p)\leq \pi/4$. 
\Lsref{lem:OptStratforECin2qPS}-\ref{lem:POmegathetap2} below are proved in \aref{app:OS&SPfor2qPS}.

\begin{figure}[t]
	\centering
	\includegraphics[scale=0.32]{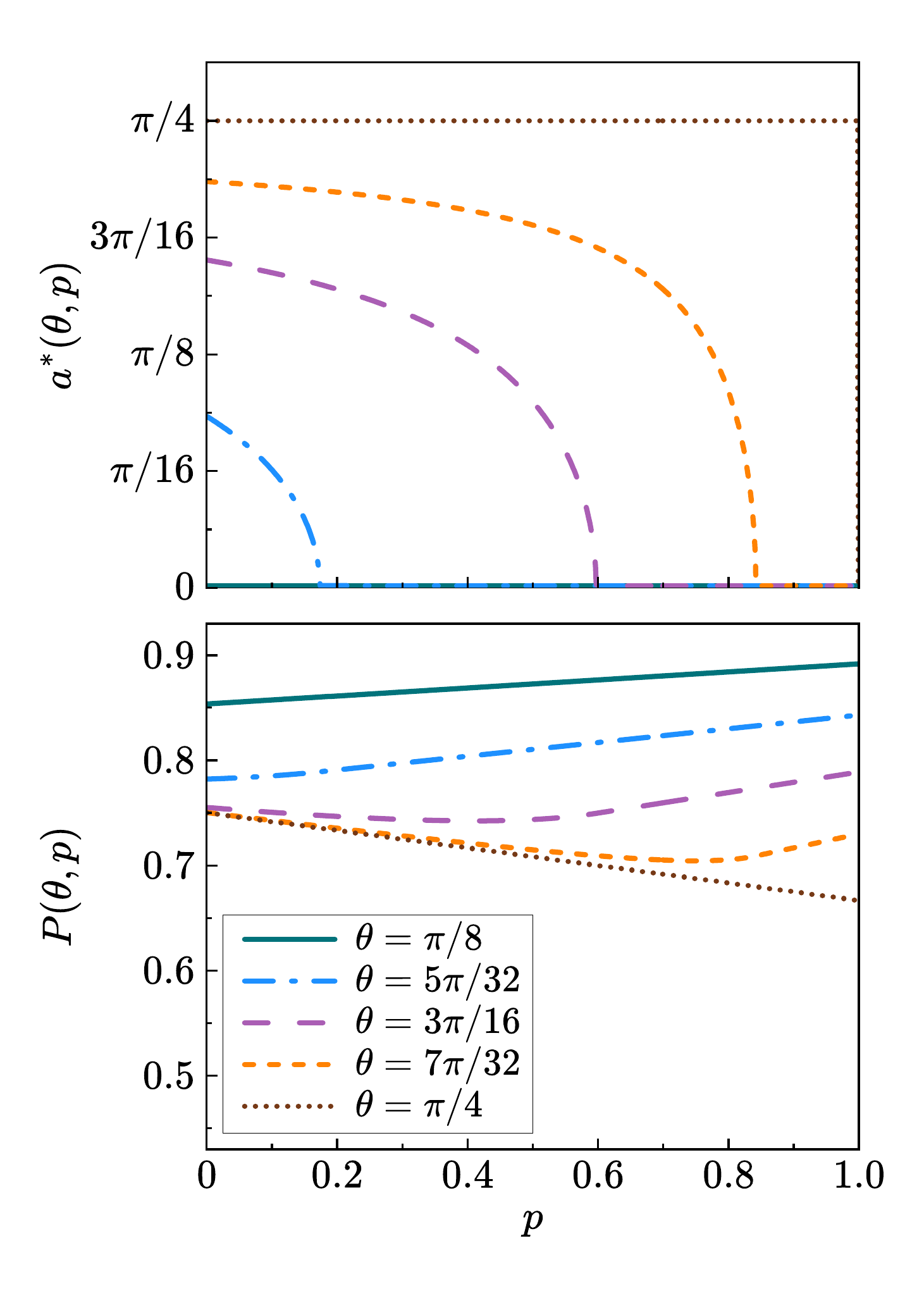}
	\caption{\label{fig:SepProbthetap}The function $a^*(\theta,p)$ and the separation probability $P(\theta,p)=P_{\Omega(\theta,p)}(\Psi_\theta)$ for certifying the entanglement of $|\Psi_\theta\rangle$. For a given value of $\theta$, $P(\theta,p)$ is convex in $p$ and has a unique minimizer in $p$. }
\end{figure}

\begin{lemma}\label{lem:OptStratforECin2qPS}
	Suppose $0< \theta\leq \pi/4$. Then, for some $0\leq p\leq 1$, $\Omega(\theta,p)$ is an optimal separable verification operator of the target state $|\Psi_\theta\>$, that is,
	\begin{equation}\label{eq:SepProb2qubitthetap}
		P_\sep(\Psi_\theta)=\min_{0\leq p\leq 1} P(\theta,p). 
	\end{equation}
\end{lemma}

\begin{lemma}\label{lem:POmegathetap}
	Suppose $0< \theta\leq \pi/4$ and $0\leq p\leq 1$. Then
	\begin{align}\label{eq:POmegathetap}
		&P(\theta,p)=\max_{0\leq a\leq \pi/2} \tr\left[\Omega(\theta,p)\rho_a^{\otimes 2}\right]= \tr\left[\Omega(\theta,p)\rho_{a^*}^{\otimes 2}\right]
		\nonumber\\	
		&=\begin{cases}
			\cos^2\theta+\frac{p\cos\theta\sin^3\theta}{1+\cos\theta\sin\theta} & 0<\theta<\pi/4,\ q(\theta)\leq p\leq 1,\\[1ex]
			\frac{9-p}{12} & \theta=\pi/4,
		\end{cases}		
	\end{align}
	where $\rho_a=|\psi_a\>\<\psi_a|$ and $|\psi_a\>=\cos a|0\>+\sin a|1\>$;  in addition, $P(\theta,p)$ is strictly increasing in $p$ for $p\in [q(\theta),1]$.	If  $p<1$ or $\theta<\pi/4$, then the maximum over $a$ is attained iff $a=a^*$. 
	If  $\arctan(1/2) < \theta<\pi/4$, then  $P(\theta,p)$ is strictly convex in $p$ for $p\in [0,q(\theta)]$.
\end{lemma}
The dependence of $P(\theta,p)$ on $\theta$ and $p$ is illustrated in \fref{fig:SepProbthetap}. 
As a simple corollary of \eref{eq:thetaFun} and \lref{lem:POmegathetap},  the separation probability $P_{\Omega_0}(\Psi_\theta)=P(\theta,0)$ reads
\begin{equation}\label{eq:POmega2}
\!\!	P_{\Omega_0}(\Psi_\theta)=
	\begin{cases}
		\cos^2\theta &\! 0<\theta\leq\arctan(1/2), \\
		\frac{3\cos^2\theta\sin^2\theta}{4\cos\theta\sin\theta-1} &\!\arctan(1/2)<\theta\leq\pi/4,
	\end{cases}
\end{equation}
as illustrated in \fref{fig:2q}.
By virtue of \lsref{lem:OptStratforECin2qPS} and~\ref{lem:POmegathetap}, we can further determine the separation probability $P_\sep(\Psi_\theta)$ and identify an optimal separable verification operator. Let $\theta^*$ be the unique root of the following equation:
\begin{equation}\label{eq:theta*}
	17 - 9 \cos(4 \theta) - 25 \sin(2 \theta) + 3\sin(6 \theta)=0,\quad \theta\in [0,\pi/4];
\end{equation}
note that $\theta^*>\arctan(1/2)$ and $\theta^*\approx 0.51095$.

\begin{figure}[t]
	\centering
	\includegraphics[scale=0.318]{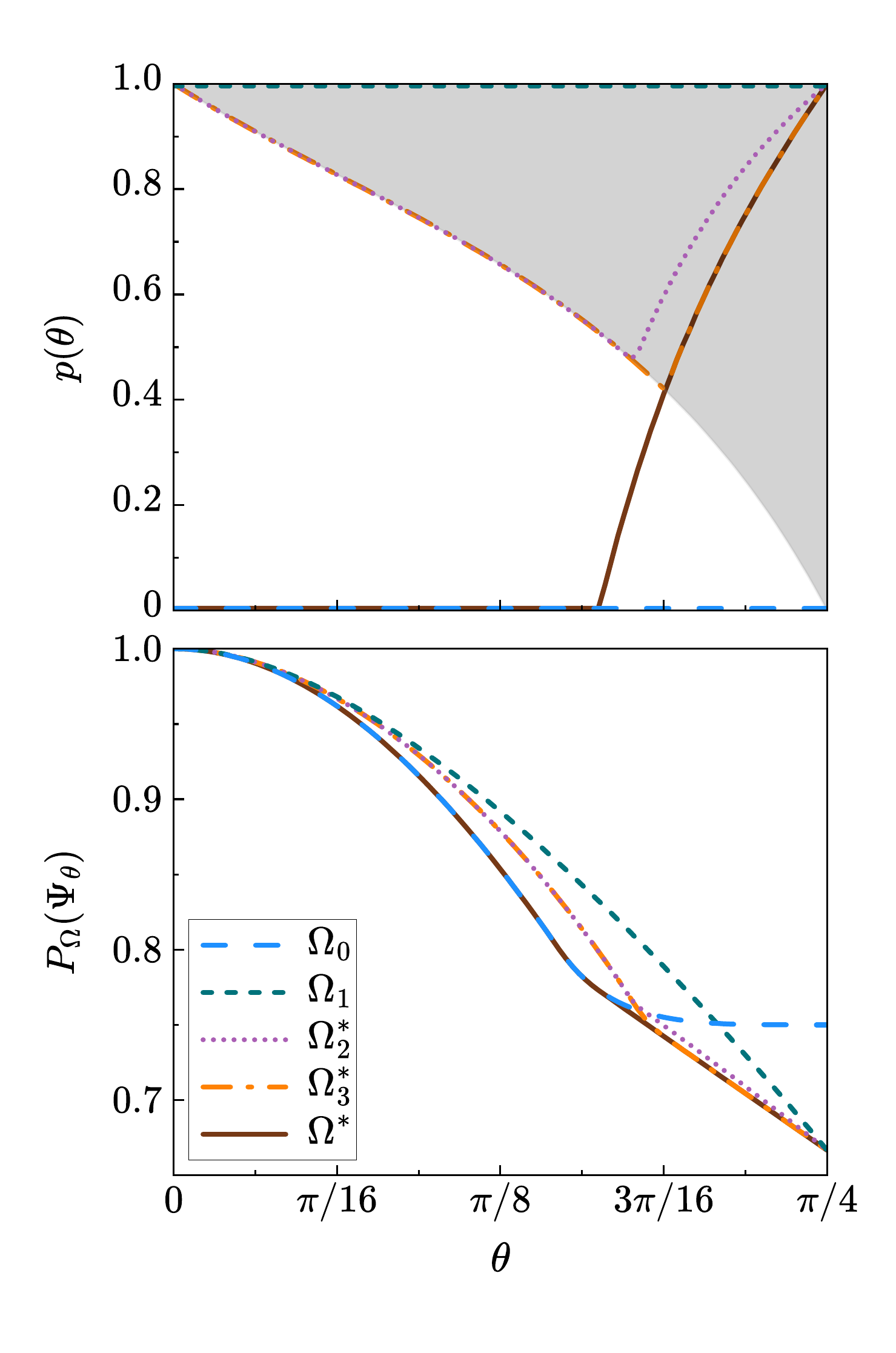}
	\caption{Several verification strategies $\Omega(\theta, p(\theta))$ of $|\Psi_\theta\>$ encoded by $\theta$ and $p(\theta)$  according to \eref{eq:Omegathetap} (upper plot) and their separation probabilities (lower plot). Here $\Omega^*$, $\Omega_2^*$, and $\Omega_3^*$ are shorthands of $\Omega(\theta,p^*(\theta))$, $\Omega(\theta,p_2^*(\theta))$, and $\Omega(\theta,p_3^*(\theta))$, respectively. Verification stategies associated with the shaded region in the upper plot can be realized by LOCC.}
	\label{fig:2q}
\end{figure}

\begin{lemma}\label{lem:POmegathetap2}
	Suppose $0< \theta\leq \pi/4$ and $0\leq p\leq 1$. Then the separation probability $P(\theta,p)$ is convex in $p$ and has a unique minimizer $p^*(\theta)$, which is contained in $[0,q(\theta)]$. The  probability $P(\theta,p)$ is strictly decreasing (increasing) in $p$ for $p\in [0,p^*(\theta)]$ $(p\in [p^*(\theta),1])$. In addition,
	\begin{equation}\label{eq:p*}
		p^*(\theta)=\begin{cases}
			0& 0<\theta\leq\theta^*,\\[0.5ex]			
			1 &\theta=\pi/4. 
		\end{cases}
	\end{equation}	
	If $\theta^*<\theta<\pi/4$, then $p^*(\theta)$ coincides with the unique zero  of the partial derivative $\partial P(\theta,p)/\partial p$ over $p \in (0,q(\theta))$.  	 
\end{lemma}

As a simple corollary of \lsref{lem:OptStratforECin2qPS}-\ref{lem:POmegathetap2} and \eref{eq:POmega2},  the following theorem determines the separation probability $P_\sep(\Psi_\theta)$. The variations 
of $p^*(\theta)$ and   $P_\sep(\Psi_\theta)$ are  illustrated in \fref{fig:2q} together with several related quantities.

\begin{theorem}\label{thm:OptStrat&SepProb2qPS}
	Suppose $0<\theta\leq\pi/4$. Then $\Omega(\theta,p^*(\theta))$ is an optimal separable strategy 
	for certifying the entanglement of $|\Psi_\theta\>$, that is, 
	$P_\sep(\Psi_\theta)=P(\theta, p^*(\theta))$. In addition,	
	\begin{equation}\label{eq:SepProb2qPS12}
	\!\!	P_\sep(\Psi_\theta)=\begin{cases}
			\cos^2\theta &0<\theta\leq\arctan(1/2), \\[0.5ex]
			\frac{3\cos^2\theta\sin^2\theta}{4\cos\theta\sin\theta-1} &\arctan(1/2)<\theta\leq\theta^*,\\[0.5ex]			
			2/3 &\theta=\pi/4.
		\end{cases}
	\end{equation}
\end{theorem} 
Thanks to \lref{lem:POmegathetap2} and \thref{thm:OptStrat&SepProb2qPS}, the optimal separable strategy for certifying the entanglement of $|\Psi_\theta\>$ is unique if we focus on strategies of the form $\Omega(\theta,p)$. Notably,
$\Omega_0$ is optimal when $0<\theta\leq\theta^*$, while $\Omega_1$ is optimal only when $\theta=\pi/4$, although it is an optimal separable strategy  for  QSV whenever $0<\theta\leq \pi/4$. In general, the optimal separable strategy for entanglement certification and the counterpart for QSV are different except when $\theta=\pi/4$.
Moreover, when $0<\theta\leq\arctan(1/2)$, the separable strategy $\Omega_0$ even achieves the same separation probability as the globally optimal strategy $\Omega=|\Psi_\theta\>\<\Psi_\theta|$, which is nonseparable. This result is in sharp contrast with the counterpart in QSV, in which the maximum spectral gap cannot be attained by separable operations for any entangled pure state.

In the rest of this section, we turn to certification strategies for $|\Psi_\theta\>$ based on LOCC. Now, the problem is much more tricky because it is not clear what strategies of the form $\Omega(\theta,p)$ can be realized by LOCC. Nevertheless,  we can  pinpoint a specific parameter region in which  this is the case. Note that, if $\Omega(\theta,p)$ can be realized by LOCC, then $\Omega(\theta,p')$ for $p'\in [p,1]$ can also  be realized by LOCC given that $\Omega(\theta,1)$ can be realized by LOCC according to Wang and Hayashi \cite{WangH19}. In addition, we can prove the following proposition  as shown in \aref{app:OmegaWH}.  
\begin{proposition}\label{pro:OmegathetapLC}
	Suppose $0<\theta\leq \pi/4$ and $\tp(\theta)\leq p\leq 1$ with 
	\begin{equation}\label{eq:tptheta}
		\tp(\theta):=	1-\frac{\tan\theta}{\cos^2\theta+\cos\theta\sin\theta}.
	\end{equation}	
	Then $\Omega(\theta,p)$ can be realized by LOCC.
\end{proposition}
Note that $\tp(\theta)$ is strictly decreasing in $\theta$, $\tp(\pi/4)=0$, and $\lim_{\theta\to 0} \tp(\theta)=1$  as illustrated in \fref{fig:2q}.  Thanks to \pref{pro:OmegathetapLC}, we can construct two efficient  strategies for certifying the entanglement of $|\Psi_\theta\>$ based on LOCC. Let  $\theta_2^*=\arctan[(\sqrt{5}-1)/2]$ and define
\begin{align}
	p_2^*(\theta)&:=\max \{\tp(\theta),q(\theta)\}=\begin{cases}
		\tp(\theta) & 0< \theta\leq \theta_2^*,\\
		q(\theta)  & \theta_2^*< \theta\leq 1,
	\end{cases}\\
	p_3^*(\theta)&:=\max \{\tp(\theta),p^*(\theta)\}.
\end{align}
Then 
\begin{align}
	0\leq p^*(\theta)\leq 	p_3^*(\theta)\leq 	p_2^*(\theta)\leq 1,
\end{align}
given that $0\leq \tp(\theta)\leq 1$ and $0\leq p^*(\theta)\leq q(\theta)\leq1$. 
So the corresponding  strategies $\Omega(\theta, p_2^*(\theta))$ and $\Omega(\theta, p_3^*(\theta))$ can be realized by LOCC. By virtue of \lref{lem:POmegathetap} we can derive an explicit formula for the separation probability of $\Omega(\theta, p_2^*(\theta))$:
\begin{align}
	&P(\theta, p_2^*(\theta))=\cos^2\theta+\frac{p_2^*\cos\theta\sin^3\theta}{1+\cos\theta\sin\theta}\nonumber\\
	&=\begin{cases}
	\cos^2\theta+\frac{(1-\tan\theta)\tan\theta\sin^2\theta}{1+\tan\theta} & 0<\theta\leq \theta_2^*,\\[1ex]
	\cos^2\theta+\frac{(2\tan\theta-1)\sin^2\theta}{(\tan\theta+2)\tan\theta} & \theta_2^*<\theta\leq \pi/4.
	\end{cases}
\end{align}
In conjunction with \lref{lem:POmegathetap2} we can deduce that
\begin{align}
	&P_\sep(\Psi_\theta)= P(\theta, p^*(\theta))\leq P(\theta, p_3^*(\theta))\leq P(\theta, p_2^*(\theta))\nonumber\\
	&\leq P(\theta, 1)= P_{\Omega_1}(\Psi_\theta)
	=\frac{\cos^2 \theta+\cos\theta\sin\theta}{1+\cos\theta\sin\theta}.
\end{align} 
This equation means the strategy $\Omega(\theta, p_3^*(\theta))$ is no worse than $\Omega(\theta, p_2^*(\theta))$, which in turn is no worse than $\Omega_1$, although $\Omega_1$ is the optimal local strategy for QSV. Actually, both strategies $\Omega(\theta, p_2^*(\theta))$ and $\Omega(\theta, p_3^*(\theta))$ are nearly optimal as illustrated in \fref{fig:2q}. This observation also shows that  the optimal local strategy for entanglement certification is in general different from the counterpart for QSV except when $\theta=\pi/4$, which echoes the result on separable operations as mentioned above.

Finally, numerical calculation shows that $\tp(\theta)\leq p^*(\theta)$, which means $p_3^*(\theta)=p^*(\theta)$ when $\theta_3^*\leq \theta\leq \pi/4$ for some threshold $\theta_3^*\approx 0.59079$. In this case, the concurrence of the target state $|\Psi_\theta\>$ satisfies
\begin{equation}
	C(\Psi_\theta)=2\sin\theta\cos\theta\geq 2\sin\theta_3^*\cos\theta_3^*\approx 0.92521.
\end{equation}
So the optimal separable entanglement certification strategy $\Omega(\theta,p^*(\theta))$ for $|\Psi_\theta\>$ can be realized by LOCC when $\theta_3^*\leq \theta\leq \pi/4$, that is, $C(\Psi_\theta)\geq 2\sin\theta_3^*\cos\theta_3^*$, which means the target state $|\Psi_\theta\>$ has sufficiently high entanglement.  

\section{Summary}\label{sec:Summary}

Inspired by QSV, we proposed a simple and general framework for certifying HDE under restricted operations and introduced the concept of separation probabilities.  As concrete examples, separation probabilities associated with the sets $\caS_\sep$, $\caS_r$, and $\caS_{\caE_r}(E)$ were discussed in detail. On this basis, we showed that HDE in general bipartite pure states, including the maximally entangled states in particular,  can be certified efficiently using LOCC. Notably, the sample cost for certifying a given degree of HDE even decreases monotonically with the local dimensions. For a general two-qubit pure state, we constructed an optimal entanglement certification strategy based on separable operations and clarified the properties of the separation  probabilities. In addition, we showed that this optimal strategy can be realized by LOCC when the target state has sufficiently high entanglement. Our work offers a new perspective on and a practical approach for HDE certification. The basic idea may also find applications in certifying many other important resources, such as coherence and nonstabilizerness, which deserves further studies.

\section*{Acknowledgment}
This work is supported by Shanghai Science and Technology Innovation Action Plan (Grant No.~24LZ1400200), Shanghai Municipal Science and Technology Major Project (Grant No.~2019SHZDZX01), and the National Key Research and Development Program of China (Grant No.~2022YFA1404204).

\bibliography{references.bib}

\clearpage

\onecolumngrid

\tableofcontents

\bigskip

In this appendix, we prove the results presented in the main text, including \psref{pro:ErLUB}-\ref{pro:ErLip}, \ref{pro:SpectralGapDimInd}-\ref{pro:QSVopt}, \ref{pro:SepProbDimInd} and \ref{pro:OmegathetapLC}, \thsref{thm:SepProb&Major}-\ref{thm:ConcSepProbforSNCinGBPS}, and \Lsref{lem:OptStratforECin2qPS}-\ref{lem:POmegathetap2}. Following the notation in the main text, $\caH_{AB}=\caH_A\otimes \caH_B$ denotes a bipartite Hilbert space with local dimensions $d=d_A, d_B$ ($d_A\leq d_B$) and total dimension $D=d_Ad_B=dd_B$. In addition, $|\Phi\>$  and $|\Psi_\theta\>$ denote the maximally entangled state defined in \eref{eq:MES} and the two-qubit pure state defined in \eref{eq:2qPS}, respectively, while $|\Psi\>$ denotes a general bipartite pure state.

\clearpage

\appendix

\section{\label{app:ErProofs}Proofs of \psref{pro:ErLUB}-\ref{pro:ErLip}}

\begin{proof}[Proof of \pref{pro:ErLUB}]
	First, suppose $\sigma=|\Psi\>\<\Psi|$ is a pure state with Schmidt spectrum $\{s_j\}_{j=0}^{d-1}$, which is arranged in nonincreasing order, that is, $s_0\geq s_1\geq \cdots\geq s_{d-1}$. Then $\caE_r(\sigma)=\caE_r(\Psi)=\sum_{j=r}^{d-1}s_j$ satisfies the inequalities in \eref{eq:ErLUB}, that is, $0\leq \caE_r(\sigma)\leq(d-r)/d$. The lower bound is saturated iff $s_r= s_{r+1}= \cdots= s_{d-1}=0$, that is, $\SN(\sigma)\leq r$; the upper bound is saturated iff $s_0= s_1= \cdots= s_{d-1}=1/d$, in which case $|\Psi\>$ is maximally entangled. 
	
	Next, we turn to a general density operator $\sigma$ in $\caD(\caH_{AB})$. Now, \eref{eq:ErLUB} still holds because $\caE_r$ is defined via the convex-roof construction. If $\SN(\sigma)\leq r$, then $\sigma$ can be expressed as a convex sum of states in $\caS_r$, which means $\caE_r(\sigma)=0$. Conversely, if  $\caE_r(\sigma)=0$, then $\sigma$ can  be expressed as a convex sum of states in $\caS_r$, which means  $\SN(\sigma)\leq r$. If $\sigma$ is maximally entangled, then every pure state $|\Upsilon\>$ in its support is maximally entangled, that is, $\caE_r(\Upsilon)=(d-r)/d$. Therefore, $\caE_r(\sigma)=(d-r)/d$ and the upper bound in \eref{eq:ErLUB} is saturated. Conversely, if $\caE_r(\sigma)=(d-r)/d$, then  every pure state $|\Upsilon\>$ in the support of $\sigma$ satisfies $\caE_r(\Upsilon)=(d-r)/d$ and is thus maximally entangled.  Therefore, $\sigma$ is also  maximally entangled. 
\end{proof}

\begin{proof}[Proof of \pref{pro:FidelityUB}]
	Expand $|\Psi\>$ and $|\Upsilon\>$ in the computational basis:
	\begin{align}
		|\Psi\>=\sum_{j=0}^{d-1}\sum_{k=0}^{d_B-1} A_{jk}|jk\>,\quad |\Upsilon\>=\sum_{j=0}^{d-1}\sum_{k=0}^{d_B-1} B_{jk}|jk\>,
	\end{align}
	where $A$ and $B$ are the coefficient matrices. 
	Then $\{s_j\}_{j=0}^{d-1}$ and $\{t_j\}_{j=0}^{d-1}$ are the  singular value spectra of $A$ and $B$, respectively. In addition, by virtue of von Neumann's trace theorem in matrix analysis \cite{HornMA85} we can deduce that
	\begin{align}
		|\< \Psi|\Upsilon \>|=\left|\tr\left(A^\dag B\right)\right|\leq \sum_{j=0}^{d-1}\sqrt{s_j t_j},
	\end{align}
which  confirms the inequality in \eref{eq:FidelityUB}. In addition, this inequality is indeed saturated when $|\Psi\>=\sum_{j=0}^{d-1}\sqrt{s_j}\lsp |jj\>$ and $|\Upsilon\>=\sum_{j=0}^{d-1}\sqrt{t_j}\lsp |jj\>$, which completes the proof of \pref{pro:FidelityUB}.
\end{proof}

\begin{proof}[Proof of \pref{pro:FPsiSk}]
	The first equality in \eref{eq:FPsiSk} holds because $\caS_r$ is the convex hull of $\tcaS_r$. By virtue of \pref{pro:FidelityUB}  the second equality in \eref{eq:FPsiSk} can be proved as follows:
	\begin{align}
		F\left(\Psi,\tcaS_r\right)=\max_{\Upsilon\in \tcaS_r}|\<\Psi|\Upsilon\>|^2
		=\max_{t_j\geq 0,\, \sum_{j=0}^{r-1} t_j=1}\left(\sum_{j=0}^{r-1}\sqrt{s_j t_j}\right)^2=\sum_{j=0}^{r-1} s_j=1-\caE_r(\Psi).
	\end{align}

	Next, we turn to \eref{eq:FPsiSkE}. 
	If $E\geq E'=\caE_r(\Psi)=\sum_{j=r}^{d-1}s_j$, then $\Psi\in \tser(E)$ and  $F(\Psi,\tser(E))=1$. If instead $0\leq E<E'$, then 
	\begin{align}
		F\left(\Psi,\tser(E)\right)=\max_{\Upsilon\in \tser(E)}|\<\Psi|\Upsilon\>|^2
		=\max_{t_j\geq 0,\, \sum_{j=r}^{d-1} t_j\leq E}\left(\sum_{j=0}^{d-1}\sqrt{s_j t_j}\right)^2=\bigl[\sqrt{E' E}+\sqrt{(1-E')(1-E)}\lsp\bigr]^2.
	\end{align}	
	In both cases, the second equality in \eref{eq:FPsiSkE} holds. Now it is straightforward to verify that $F(\Psi,\tser(E))$ is nondecreasing and concave in $E$. Note that the monotonicity of $F(\Psi,\tser(E))$ also holds by definition. 
	
	To prove the first equality in \eref{eq:FPsiSkE}, suppose $F(\Psi,\ser(E))=\<\Psi|\sigma|\Psi\>$ with $\sigma\in \ser(E)$, which means $\caE_r(\sigma)\leq E$. Let $\sigma=\sum_l p_l|\Psi_l\>\<\Psi_l|$ be an optimal decomposition of $\sigma$ into pure states such that $\caE_r(\sigma)=\sum_l p_l \caE_r(\Psi_l)$. Then 
	\begin{align}
		F\left(\Psi,\ser(E)\right)=\<\Psi|\sigma|\Psi\>=\sum_l p_l|\<\Psi|\Psi_l\>|^2\leq \sum_l p_l F\left(\Psi,\tser(\caE_r(\Psi_l))\right)
		\leq  F\left(\Psi,\tser(\caE_r(\sigma))\right)	\leq F\left(\Psi,\tser(E)\right), 
	\end{align}	
	where the first inequality holds because $\Psi_l \in \tser(\caE_r(\Psi_l))$ by definition, while the second and third inequalities hold because $F(\Psi,\tser(E))$ is concave and nondecreasing in $E$ as proved above. Therefore, $F(\Psi,\ser(E))=F(\Psi,\tser(E))$, which confirms the first equality in \eref{eq:FPsiSkE}, given that $F(\Psi,\ser(E))\geq F(\Psi,\tser(E))$ by definition. Consequently, $F(\Psi,\ser(E))$ is also nondecreasing and concave in $E$, which completes the proof of  \pref{pro:FPsiSk}. 
\end{proof}

\begin{proof}[Proof of \pref{pro:ErLip}]
	Let $E=\caE_r(\Psi)$, $E'=\caE_r(\Upsilon)$, $a=\arcsin\sqrt{E}$, and $b=\arcsin\sqrt{E'}$; then $0\leq E, E'\leq (d-r)/d$ and $0\leq \sin(a+b)\leq 1$. If in addition $r>d/2$, then $\sin(a+b)\leq 2\sqrt{r(d-r)}/d$. Therefore, $0\leq \sin(a+b)\leq \ell(r,d)$ for $r\in [d-1]$.

	On the other hand, we have
	\begin{align}
		|E-E'|&=|\caE_r(\Psi)-\caE_r(\Upsilon)|=\left| \sin^2 a-\sin^2 b \right|=\left| \sin^2a \cos^2b-\sin^2b\cos^2a \right| =\left| \sin(a+b)\sin(a-b) \right| \nonumber\\
		&=\left| 2\sin(a+b)\sin\left( \frac{a-b}{2} \right)\cos\left( \frac{a-b}{2} \right) \right|\leq \left| 2\sin(a+b)\sin\left( \frac{a-b}{2} \right)\right|\leq 2\ell(r,d)  \left| \sin\left( \frac{a-b}{2} \right)\right|.
	\end{align}
	Furthermore, by virtue of \pref{pro:FPsiSk} we can deduce that 
	\begin{gather}
		|\<\Psi|\Upsilon\>|\leq \sqrt{E' E}+\sqrt{(1-E')(1-E)}=\sin a\sin b+\cos a\cos b=\cos(a-b),\\
		\sqrt{2-2|\<\Psi|\Upsilon\>|}\geq \sqrt{2-2\cos(a-b)}=2\left|\sin\left( \frac{a-b}{2} \right)\right|.
	\end{gather}
	The above equations together imply that 
	\begin{align}
		|\caE_r(\Psi)-\caE_r(\Upsilon)|&\leq \ell(r,d)\sqrt{2-2|\<\Psi|\Upsilon\>|}\leq\ell(r,d) \| |\Psi\>-|\Upsilon\>\|_2,
	\end{align}
	which confirms \eref{eq:ErLip} and completes the proof of \pref{pro:ErLip}. Here the last inequality  follows from the equation below:
	\begin{align}
		\||\Psi\>-|\Upsilon\>\|_2^2=2-\<\Psi|\Upsilon\>-\<\Upsilon|\Psi\>\geq 2-2|\<\Psi|\Upsilon\>|. 
	\end{align}
\end{proof}

\section{\label{app:QSVProofs}Proofs of \psref{pro:SpectralGapDimInd}-\ref{pro:QSVopt}}
\Pref{pro:SpectralGapDimInd} is a simple corollary of the following lemma.
\begin{lemma}\label{lem:nuOmegaEmbed}
	Suppose $\caH_{AB}=\caH_A\otimes \caH_B$ and $\caH_{AB}'=\caH_A'\otimes \caH_B'$ are two bipartite Hilbert spaces and the two states
	$|\Psi\>\in \caH_{AB}$ and $|\Psi'\>\in \caH_{AB}'$ have the same nonzero Schmidt coefficients (including the multiplicities). 
	Then 
	\begin{align}
		\nu_\lc (\Psi)&=\nu_\lc (\Psi'), \quad \nu_\sep (\Psi)=\nu_\sep (\Psi'), \label{eq:betaOmegaEmbed} \\
		\nu_\lc^\rmH (\Psi)&=\nu_\lc^\rmH (\Psi'), \quad \nu_\sep^\rmH (\Psi)=\nu_\sep^\rmH (\Psi'). \label{eq:betaOmegaEmbedH}
	\end{align}
\end{lemma}

\begin{proof}[Proof of \lref{lem:nuOmegaEmbed}]
	Without loss of generality we can assume that $\caH_A\leq \caH_A'$ and $\caH_B\leq \caH_B'$. If $\caH_A= \caH_A'$ and $\caH_B=\caH_B'$, then $|\Psi\>$ and $|\Psi'\>$ are equivalent under a local unitary transformation, and there is a one-to-one correspondence between local (separable)  verification operators of $|\Psi\>$ and the counterparts of $|\Psi'\>$ under the same local unitary transformation.
	The same conclusion holds if we restrict to  homogeneous verification operators. So \eqsref{eq:betaOmegaEmbed}{eq:betaOmegaEmbedH} hold as expected.

	In general, by applying a suitable local unitary transformation if necessary, we can assume that $|\Psi'\>$ is supported in $\caH_{AB}$ and is identical to $|\Psi\>$ when regarded as a pure state in $\caH_{AB}$.
	Then any local verification operator of $|\Psi\>$ is also a local verification operator of $|\Psi'\>$, which implies that $\nu_\lc (\Psi)\leq \nu_\lc (\Psi')$.
	To prove the opposite inequality, let $Q_A$ ($Q_B$) be the projector onto $\caH_A$ ($\caH_B$) as a subspace of $\caH_A'$ ($\caH_B'$) and let $Q=Q_A\otimes Q_B$. If $\Omega'$ is a local verification operator of $|\Psi'\>$, then $\Omega=Q\Omega' Q$  is a local verification operator of $|\Psi\>$ with $\nu(\Omega)\geq \nu(\Omega')$. Therefore, we have $\nu_\lc (\Psi)\geq \nu_\lc (\Psi')$, which implies the first equality in \eref{eq:betaOmegaEmbed} given the opposite inequality proved above.
	
	Next, we turn to \eref{eq:betaOmegaEmbedH}. If  $\Omega'$ is a local homogeneous verification operator of $|\Psi'\>$, then $\Omega=Q\Omega' Q$  is  a local homogeneous verification operator of $|\Psi\>$ with $\nu(\Omega)\geq \nu(\Omega')$, which implies that $\nu_\lc^\rmH (\Psi)\geq \nu_\lc^\rmH (\Psi')$. Conversely, if $\Omega$ is a local homogeneous verification operator of $|\Psi\>$, then $\Omega'=\Omega+\beta(\Omega)(\bbone'-Q)$  
	is  a local homogeneous verification operator of $|\Psi'\>$ with $\nu(\Omega ')= \nu(\Omega)$. Therefore,  $\nu_\lc^\rmH (\Psi)\leq \nu_\lc^\rmH (\Psi')$, which implies the first equality in \eref{eq:betaOmegaEmbedH} given the opposite inequality proved above. 
	
	The above reasoning still applies when local operations are replaced by separable operations, so the second inequalities in \eqsref{eq:betaOmegaEmbed}{eq:betaOmegaEmbedH} also hold, which completes the proof of \lref{lem:nuOmegaEmbed}.
\end{proof}

\begin{proof}[Proof of \pref{pro:OmegaLCsep}]
	By construction it is clear that  $\Omega_\sep$ is  a  verification operator of $|\Psi\>$, while $\Omega_\sep^\rmH$ and $\Omega_\lc^\rmH$ are homogeneous  verification operators of $|\Psi\>$.  According to the proofs of Lemma~1 and Theorem~2 of \rcite{Owari08} (see also \rcite{VidaT99}), the two operators  $\Omega_\sep$ and $\bbone-\Omega_\sep$
	are separable. Therefore, $\Omega_\sep$ is  a separable verification operator of $|\Psi\>$. In addition, we have
	\begin{align}
		&\left(1+\sqrt{s_0s_1}\lsp\right)\Omega_\sep^\rmH=\Omega_\sep+ \sqrt{s_0s_1} \sum_{j=0}^{d-1} |jj\>\<jj|+\sum_{j,k=0,\,j\neq k}^{d-1}\left(\sqrt{s_0s_1}- \sqrt{s_js_k}\lsp\right) |jk\>\<jk|+\sqrt{s_0s_1} \sum_{j=0}^{d-1}\sum_{k=d}^{d_B-1} |jk\>\<jk|,\\
		&\left(1+\sqrt{s_0s_1}\lsp\right)\left(\bbone-\Omega_\sep^\rmH\right)=\bbone-|\Psi\>\<\Psi|
		=\bbone-\Omega_\sep+\sum_{j,k=0,\,j\neq k}^{d-1} \sqrt{s_js_k}\lsp |jk\>\<jk|,
	\end{align}
	which imply that $\Omega_\sep^\rmH$ and $\bbone-\Omega_\sep^\rmH$ are separable given that  $\Omega_\sep$ and $\bbone-\Omega_\sep$ are separable. Therefore, $\Omega_\sep^\rmH$ is  a separable homogeneous verification operator of $|\Psi\>$.
	
	Next, according to \rscite{LiHZ19}, the verification operator $\Omega_\lc^\rmH$ can be realized by LOCC. So $\Omega_\lc^\rmH$ is a local homogeneous verification operator of $|\Psi\>$. 
	
	Equations~\eqref{eq:GapSep}-\eqref{eq:GapLCH} follow from the definitions of the verification operators $\Omega_\sep$, $\Omega_\sep^\rmH$, and $\Omega_\lc^\rmH$ in Eqs.~\eqref{eq:OmegaSep}-\eqref{eq:OmegaLCH}, which completes the proof of \pref{pro:OmegaLCsep}. 
\end{proof}

\begin{proof}[Proof of \pref{pro:betaOmegaLB}]
	Thanks to \pref{pro:SpectralGapDimInd}, we can assume that $r=d=d_A=d_B$ without loss of generality.

	According to Lemma~2 of \rcite{ZhuHOVMES19},  any separable verification operator $\Omega$ of $|\Psi\>$ satisfies
	\begin{equation}
		\beta(\Omega)\geq\frac{\caE_R(\Psi)}{d_A d_B-1}=\frac{\left(\sum_j \sqrt{s_j}\right)^2-1}{r^2-1},
	\end{equation}
	which implies the first inequality in \eref{eq:betaOmegaLB1}. Here $	\caE_R(\Psi)=(\sum_j \sqrt{s_j})^2-1$ is the robustness of entanglement \cite{VidaT99}. If in addition $\Omega$ is homogeneous, then 
	\begin{equation}
		\beta(\Omega)\geq \frac{R(\Psi)}{d_A d_B+R(\Psi)}=\frac{\sqrt{s_0 s_1}}{1+\sqrt{s_0 s_1}},
	\end{equation}
	where $R(\Psi)=d_A d_B \sqrt{s_0 s_1}$ is the random robustness \cite{VidaT99}. The above inequality is saturated when $\Omega=\Omega_\sep^\rmH$ thanks to \pref{pro:OmegaLCsep}. Therefore, $\beta_\sep^\rmH=\sqrt{s_0 s_1}/(1+\sqrt{s_0 s_1}\lsp)$, which confirms  the first equality in \eref{eq:betaOmegaLB2}. The third inequality in \eref{eq:betaOmegaLB1} and the second inequality in \eref{eq:betaOmegaLB2} also follow from \pref{pro:OmegaLCsep}. The inequalities $\beta_\sep\leq \beta_\lc$ and $\beta_\sep^\rmH\leq \beta_\lc^\rmH$ in \eqsref{eq:betaOmegaLB1}{eq:betaOmegaLB2} hold by  definition. The inequality $(s_0+s_1)/(2+s_0+s_1)\leq 1/3$ holds because $2/d\leq s_0+s_1\leq 1$. This observation completes the proof of \pref{pro:betaOmegaLB}. 
\end{proof}

\begin{proof}[Proof of \pref{pro:QSVopt}]
	By assumption $|\Psi\>$   is maximally entangled, which means  $\SR(\Psi)=d=d_A$, $s_j=1/d$ for $j=0, 1,\ldots, d-1$. So the first three inequalities in \eref{eq:betaOmegaLB1} and the first two inequalities in  \eref{eq:betaOmegaLB2} are saturated, which implies \eref{eq:betaOmegaOpt}. In addition, the verification operator  $\Omega_\opt$ defined in \eref{eq:PsiOmegaOpt} is separable and local. If $\Omega=\Omega_\opt$, then it is obvious that $\beta(\Omega)=1/(d+1)$.

	Next, suppose $d=d_A=d_B$ and $\Omega$ is a separable verification operator of $|\Psi\>$ with $\beta(\Omega)=1/(d+1)$. To prove \eref{eq:PsiOmegaOpt} we can assume that $|\Psi\>=|\Phi\>$ without loss of generality. Then $\caH_A$ and $\caH_B$ are isomorphic and $|\Psi\>$ is invariant under local unitary transformations of the form $U\otimes U^*$ for $U\in \rmU(d)$, where $U^*$ denotes the complex conjugation of $U$ with respect to the Schmidt basis. Let 
	\begin{align}
		\Omega^\rmH=\int (U\otimes U^*)\Omega  (U\otimes U^*)^\dag d U,
	\end{align}
	where the integration is taken over the normalized Haar measure on $\rmU(d)$. Then $\Omega^\rmH$ is a separable homogeneous verification operator of $|\Psi\>$. In addition, we have
	\begin{align}
\tr\left(\Omega^\rmH\right)=\tr(\Omega),\quad 		\frac{1}{d+1}\leq \beta\left(\Omega^\rmH\right)\leq \beta(\Omega) =\frac{1}{d+1}, 
	\end{align}
	which implies that  $\beta\left(\Omega^\rmH\right)= \beta(\Omega)=1/(d+1)$. Therefore, $\Omega$ is necessarily homogeneous and we have
	\begin{equation}
		\Omega=\Omega^\rmH=|\Psi\>\<\Psi|+\frac{1}{d+1}(\bbone-|\Psi\>\<\Psi|),
	\end{equation} 
	which confirms \eref{eq:PsiOmegaOpt} and completes the proof of \pref{pro:QSVopt}. 
\end{proof}

\section{\label{app:SepProbProofs}Proof of \pref{pro:SepProbDimInd}}

\pref{pro:SepProbDimInd} is a simple corollary of \lref{lem:SepProbEmbed} below.

\begin{lemma}\label{lem:SepProbEmbed}
	Suppose $\caH_{AB}=\caH_A\otimes \caH_B$ and $\caH_{AB}'=\caH_A'\otimes \caH_B'$ are two bipartite Hilbert spaces and the two states
	$|\Psi\>\in \caH_{AB}$ and $|\Psi'\>\in \caH_{AB}'$ have the same nonzero Schmidt coefficients (including multiplicities).  Suppose $r\in[d-1]$, $0\leq E< \caE_r(\Psi)$, $\caS$ is one of the sets $\caS_r$ or $\ser(E)$, and $\caS'$ is the counterpart with $\caH_{AB}$ replaced by $\caH_{AB}'$.
	Then 
	\begin{align}
	P_\lc (\Psi,\caS)&=P_\lc (\Psi',\caS'), \quad  P_\sep (\Psi,\caS)	=P_\sep (\Psi',\caS'),  \label{eq:SepProbEmbed} \\		
	P_\lc^\rmH (\Psi,\caS)&=P_\lc^\rmH (\Psi',\caS'), \quad  P_\sep^\rmH (\Psi,\caS)	=P_\sep^\rmH (\Psi',\caS'). \label{eq:SepProbEmbedH}		
	\end{align}
\end{lemma}

\begin{proof}[Proof of \lref{lem:SepProbEmbed}]
	As in the proof of \lref{lem:nuOmegaEmbed} 	 we can assume that $\caH_A\leq \caH_A'$ and $\caH_B\leq \caH_B'$, then $\caS\subseteq \caS'$. 
	If $\caH_A= \caH_A'$ and $\caH_B=\caH_B'$, then $|\Psi\>$ and $|\Psi'\>$ are equivalent under  local unitary transformations, so \eqsref{eq:SepProbEmbed}{eq:SepProbEmbedH} hold as expected. 
	
	In general, by applying a suitable local unitary transformation if necessary, we can assume that $|\Psi'\>$ is supported in $\caH_{AB}$ and is identical to $|\Psi\>$ when regarded as a pure state in $\caH_{AB}$.
	Then any local verification operator $\Omega$ of $|\Psi\>$ is also a local verification operator of $|\Psi'\>$ (via the natural embedding). In addition, we have $P_\Omega(\Psi,\caS)=P_\Omega(\Psi',\caS')$ by \lref{lem:EmbedProbId} below, which implies that $P_\lc (\Psi,\caS)\geq P_\lc (\Psi',\caS')$. 
	
	To prove the opposite inequality, let $Q_A$ ($Q_B$) be the projector onto $\caH_A$ ($\caH_B$) as a subspace of $\caH_A'$ ($\caH_B'$) and let $Q=Q_A\otimes Q_B$.
	If $\Omega'$ is a local verification operator of $|\Psi'\>$, then $Q\Omega' Q$  can be regarded as a local verification operator of $|\Psi\>$. In addition, we have
	\begin{align}
	P_{Q\Omega' Q}(\Psi,\caS)=\max_{\sigma\in \caS} \tr(Q\Omega' Q\sigma )=\max_{\sigma\in \caS} \tr(\Omega' \sigma )\leq \max_{\sigma\in \caS'} \tr(\Omega' \sigma )=P_{\Omega'}(\Psi',\caS'), 
	\end{align}
	which implies that $P_\lc(\Psi,\caS)\leq P_\lc(\Psi',\caS')$. In conjunction with the opposite inequality proved above, we can deduce that $P_\lc(\Psi,\caS)= P_\lc(\Psi',\caS')$, which confirms the first equality in \eref{eq:SepProbEmbed}. The second equality in \eref{eq:SepProbEmbed} follows from a similar reasoning. 
	
	\Eref{eq:SepProbEmbedH} follows from
	\psref{pro:FPsiSk}, \ref{pro:SepProbUBQSV} and \lref{lem:nuOmegaEmbed}, which completes the proof of \lref{lem:SepProbEmbed}.
\end{proof}

In the rest of this appendix, we prove two auxiliary lemmas employed in the proof of \lref{lem:SepProbEmbed}.  
\begin{lemma}\label{lem:SepProbMonoConcave}
	Suppose $|\Psi\>\in \caH_{AB}$, $r\in[d-1]$,  $0\leq E< \caE_r(\Psi)$, and $\Omega$ is a verification operator of $|\Psi\>$. Then $P_\Omega(\Psi,\ser(E))$ and
	$P(\Psi, \ser(E))$ are nondecreasing and concave in $E$.             
\end{lemma}

\begin{proof}[Proof of \lref{lem:SepProbMonoConcave}]
	The monotonicity of $P_\Omega(\Psi,\ser(E))$ and $P(\Psi, \ser(E))$ is evident by definition.
	To prove their concavity, it suffices to prove the following inequalities:
	\begin{gather}
		P_\Omega(\Psi,\ser(E))\geq q P_\Omega(\Psi,\ser(E_1)) +(1-q) P_\Omega(\Psi,\ser(E_2)),\label{eq:POmegaConcave}\\
		P(\Psi,\ser(E))\geq q P(\Psi,\ser(E_1)) +(1-q) P(\Psi,\ser(E_2)),\label{eq:PConcave}
	\end{gather}
	assuming that  $0\leq E_1, E_2< \caE_r(\Psi)$, $0\leq q\leq 1$, and $q E_1+(1-q)E_2=E$. Suppose $P_\Omega(\Psi,\ser(E_j))=\tr(\Omega\rho_j)$  with $\rho_j\in \ser(E_j)$ for $j=1,2$. Then we have
	\begin{equation}
		q P_\Omega(\Psi,\ser(E_1)) +(1-q) P_\Omega(\Psi,\ser(E_2))=q\tr(\Omega\rho_1)+(1-q)\tr(\Omega\rho_2)=\tr\{\Omega[q\rho_1+(1-q)\rho_2]\}\leq P_\Omega(\Psi,\ser(E)),
	\end{equation}
		where the inequality holds because $q\rho_1+(1-q)\rho_2\in \ser(E)$ by the definition of the entanglement measure $\caE_r(\cdot)$.
	The above equation confirms \eref{eq:POmegaConcave} and the concavity of $P_\Omega(\Psi,\ser(E))$.
Next, suppose $\Omega'$ is an optimal strategy for certifying the target state $|\Psi\>$ against $\ser(E)$. Then 	
	\begin{align}
P(\Psi,\ser(E))&=P_{\Omega'}(\Psi,\ser(E))
\geq  q P_{\Omega'}(\Psi,\ser(E_1)) +(1-q) P_{\Omega'}(\Psi,\ser(E_2))
\nonumber\\
&\geq q P(\Psi,\ser(E_1)) +(1-q) P(\Psi,\ser(E_2)),
\end{align}
where the first inequality follows from the concavity of $P_{\Omega'}(\Psi,\ser(E))$ as proved above, and the second inequality holds by definition. This observation confirms \eref{eq:PConcave} and the concavity of $P(\Psi,\ser(E))$ and completes the proof of \lref{lem:SepProbMonoConcave}.	
\end{proof}

\begin{lemma}\label{lem:EmbedProbId}
	Suppose  $\caH_{AB}=\caH_A\otimes \caH_B$ and $\caH_{AB}'=\caH_A'\otimes \caH_B'$ are two bipartite Hilbert spaces with $\caH_A\leq \caH_A'$ and $\caH_B\leq \caH_B'$,  $\Omega$ is a verification operator of $|\Psi\>\in \caH_{AB}$,  $\caS$ is one of the sets $\caS_r$ or $\ser(E)$ with $r\in[d-1]$ and $0\leq E< \caE_r(\Psi)$, and $\caS'$ is the counterpart with $\caH_{AB}$ replaced by $\caH_{AB}'$. Then 
	\begin{align}\label{eq:EmbedProbId}
		P_\Omega(\Psi,\caS)=P_\Omega(\Psi,\caS').
	\end{align}
\end{lemma}

\begin{proof}[Proof of \lref{lem:EmbedProbId}]
	Let $P_\Omega=P_\Omega(\Psi,\caS)$ and $P_\Omega'=P_\Omega(\Psi,\caS')$; 	
	then $P_\Omega'\geq P_\Omega\geq 0$ given that $\caS\subseteq \caS'$. If  $P_\Omega'=0$, then $P_\Omega= P_\Omega'= 0$. Otherwise,
	let $Q_A$ ($Q_B$) be the projector onto $\caH_A$ ($\caH_B$) as a subspace of $\caH_A'$ ($\caH_B'$) and let $Q=Q_A\otimes Q_B$. Suppose $P_\Omega'=\tr(\Omega\rho)$ with $\rho\in \caS'$. Let  $q=\tr(Q\rho Q)$ and $\varrho =Q\rho Q/q$ (note that $0<q\leq 1$); then 
	\begin{align}
		P_\Omega'=\tr(\Omega\rho)=q \tr(\Omega\varrho)\leq \tr(\Omega\varrho). 
	\end{align} 
	If $\caS=\caS_r$, then $\rho\in \caS_r'$, $\varrho\in \caS_r$, and the above equation implies that $P_\Omega'\leq P_\Omega$, which confirms \eref{eq:EmbedProbId} given that the opposite inequality holds by definition. 
	
	Next, suppose $\caS=\ser(E)$. Then $\rho\in \ser'(E)$ and $q\caE_r(\varrho)\leq \caE_r(\rho)\leq E$. In addition, we have
	\begin{align}
		P_\Omega'&=q \tr(\Omega\varrho)\leq q P_{\Omega}(\Psi,\ser(\caE_r(\varrho)))\leq (1-q)P_{\Omega}(\Psi,\ser(0))  +q P_{\Omega}(\Psi,\ser(\caE_r(\varrho)))\nonumber\\
		&\leq P_{\Omega}(\Psi,\ser(q\caE_r(\varrho)))\leq P_{\Omega}(\Psi,\ser(E))=P_\Omega,
	\end{align}
	given that $P_{\Omega}(\Psi,\ser(E))$ is nondecreasing and concave in $E$ by \lref{lem:SepProbMonoConcave}. In conjunction with the opposite inequality $P_\Omega'\geq P_\Omega$ mentioned above, this equation 
	implies \eref{eq:EmbedProbId} and completes the proof of \lref{lem:EmbedProbId}. 
\end{proof}

\section{\label{app:SepProbGBPSproofs}Proofs of results on general bipartite pure states}

In this appendix we prove \thsref{thm:SepProb&Major}-\ref{thm:ConcSepProbforSNCinGBPS}, which are tied to  entanglement certification of general bipartite pure states.

\subsection{Proof of \thref{thm:SepProb&Major}}

\begin{proof}[Proof of \thref{thm:SepProb&Major}]
	Suppose $|\Upsilon\>$ is another quantum state in $\caH_{AB}$. According to \pref{pro:Majorization},  $|\Psi\>$ is majorized by $|\Upsilon\>$  iff  there exists a local channel $\Lambda$ that can transform $|\Psi\>$ into $|\Upsilon\>$, that is, $\Lambda(|\Psi\>\<\Psi|)=|\Upsilon\>\<\Upsilon|$. Here we assume that both conditions hold in the following analysis.

	Let $\Omega_\Upsilon$ be an arbitrary verification operator of $|\Upsilon\>$; then $0\leq \Omega_\Upsilon\leq \bbone$ and $\tr\left(\Omega_{\Upsilon} |\Upsilon\>\<\Upsilon|\right)=1$. Let $\Omega_{\Psi}=\Lambda^\dagger(\Omega_{\Upsilon})$; then 
	\begin{align}
		0\leq \Omega_{\Psi}\leq \bbone,\quad 	\tr\left(\Omega_{\Psi}|\Psi\>\<\Psi|\right)
		=\tr\left[\Omega_{\Upsilon}\Lambda(|\Psi\>\<\Psi|)\right]=\tr\left(\Omega_{\Upsilon}|\Upsilon\>\<\Upsilon|\right)=1.
	\end{align}
	Therefore, $\Omega_{\Psi}$ is a verification operator of $|\Psi\>$. If $\Omega_{\Upsilon}$ can be realized by LOCC (separable operations), then $\Omega_{\Psi}$ can also be realized by LOCC (separable operations). 
	Moreover, the separation probability achieved by $\Omega_{\Psi}$ can be bounded from above as follows: 
	\begin{equation}
		P_{\Omega_{\Psi}}(\Psi,\caS)=	\max_{\sigma\in\caS}\tr(\Omega_{\Psi} \sigma)=\max_{\sigma\in\caS}\tr\left[\Omega_{\Upsilon} \Lambda(\sigma)\right]\leq \max_{\sigma\in\caS}\tr(\Omega_{\Upsilon} \sigma)=P_{\Omega_{\Upsilon}}(\Upsilon,\caS). 
	\end{equation}
	Here the inequality holds because $\Lambda(\sigma)\in \caS$ whenever $\sigma\in \caS$. Therefore, 
	\begin{equation}
		P_\lc(\Psi,\caS)\leq P_\lc(\Upsilon,\caS),\quad P_\sep(\Psi,\caS)\leq P_\sep(\Upsilon,\caS), 
	\end{equation}
	which means $P_\lc(\Psi,\caS)$ and $P_\sep(\Psi,\caS)$ are Schur convex in $|\Psi\>$ and do not decrease if $|\Psi\>$ is subjected to LOCC.

The equality in \eref{eq:SepProbPsiPhi} holds because the local homogeneous verification operator $\Omega_{\opt}$ in \pref{pro:QSVopt} is optimal among separable verification operators; the first inequality holds because the maximally entangled state $|\Phi\>$ is majorized by any pure state in $\caH_{AB}$, and the second inequality holds by definition.
\Eref{eq:SepProbPsiPhiH} is a simple corollary of \eref{eq:SepProbPsiPhi} given that $P_\lc^\rmH(\Psi,\caS)\geq P_\sep^\rmH(\Psi,\caS)\geq P_\sep(\Psi,\caS)$ by definition. This observation completes the proof of \thref{thm:SepProb&Major}.	
\end{proof}

\subsection{Proof of \thref{thm:SepProbGBPSLUB}}
\begin{proof}[Proof of \thref{thm:SepProbGBPSLUB}]
By virtue of \psref{pro:FPsiSk} and \ref{pro:SepProbUBQSV} we can deduce that
\begin{gather}
	P_\sep(\Psi,\caS_r)\geq F(\Psi,\caS_r)=1-\caE_r(\Psi), \quad 
	P_\sep(\Psi,\ser(E))\geq F(\Psi,\ser(E))=f_r(\Psi,E),
\end{gather}
which, together with the definitions in \eref{eq:PsepLBUB}, confirm the first inequalities in \eqsref{eq:SepProbGBPSLUB}{eq:SepProbErGBPSLUB}. By virtue of  \psref{pro:FPsiSk}, \ref{pro:betaOmegaLB}, and \ref{pro:SepProbUBQSV} we can deduce that
\begin{gather}
    P_\sep^\rmH(\Psi,\caS_r)=\nu_\sep^\rmH(\Psi) F(\Psi,\caS_r)+\beta_\sep^\rmH(\Psi)=\frac{[1-\caE_r(\Psi)]+\sqrt{s_0 s_1}}{1+\sqrt{s_0 s_1}}=1-\frac{\caE_r(\Psi)}{1+\sqrt{s_0 s_1}}, \\
	P_\sep^\rmH(\Psi,\ser(E))=\nu_\sep^\rmH(\Psi) F(\Psi,\ser(E))+\beta_\sep^\rmH(\Psi)=\frac{f_r(\Psi,E)+\sqrt{s_0 s_1}}{1+\sqrt{s_0 s_1}},
\end{gather}
which confirm  the equalities in \eqsref{eq:SepProbGBPSLUBH}{eq:SepProbErGBPSLUBH}.

The second inequalities in \eqsref{eq:SepProbGBPSLUB}{eq:SepProbErGBPSLUB} and the first inequalities in \eqsref{eq:SepProbGBPSLUBH}{eq:SepProbErGBPSLUBH} hold by definition.

The first upper bounds for $P_\lc(\Psi,\caS_r)$, $P_\lc^\rmH(\Psi,\caS_r)$, $P_\lc(\Psi,\ser(E))$, and $P_\lc^\rmH(\Psi,\ser(E))$ correspond to the separation probabilities achieved by the following local homogeneous verification operator [see \eref{eq:OmegaLCH}]:
\begin{equation}
\Omega=|\Psi\>\<\Psi|+\frac{s_0+s_1}{2+s_0+s_1}(\bbone-|\Psi\>\<\Psi|).
\end{equation}
In conjunction with \psref{pro:FPsiSk} and \ref{pro:PassProbUB} we can deduce that
\begin{gather}
P_\lc(\Psi,\caS_r)\leq P_\lc^\rmH(\Psi,\caS_r)\leq P_\Omega(\Psi,\caS_r)=\nu(\Omega) F(\Psi,\caS_r)+\beta(\Omega)\leq \frac{2[1-\caE_r(\Psi)]+s_0+s_1}{2+s_0+s_1}=1-\frac{2\caE_r(\Psi)}{2+s_0+ s_1}, \\
P_\lc(\Psi,\ser(E))\leq P_\lc^\rmH(\Psi,\ser(E))\leq P_\Omega(\Psi,\ser(E))=\nu(\Omega) F(\Psi,\ser(E))+\beta(\Omega)\leq \frac{2f_r(\Psi,E)+s_0+s_1}{2+s_0+s_1},
\end{gather}
which, together with the definitions in \eref{eq:PsepLBUB}, confirm the third inequalities in \eqsref{eq:SepProbGBPSLUB}{eq:SepProbErGBPSLUB} and the second inequalities in \eqsref{eq:SepProbGBPSLUBH}{eq:SepProbErGBPSLUBH}.

Finally, by virtue of the inequalities $f_r(\Psi,E)\geq f_r(\Psi,0)=1-\caE_r(\Psi)\geq s_0$ we can deduce that
\begin{gather}
1-\frac{2\caE_r(\Psi)}{2+s_0+ s_1}\leq 1-\frac{\caE_r(\Psi)}{1+s_0}=\frac{1-\caE_r(\Psi)}{1+s_0}+\frac{s_0}{1+s_0}\leq \frac{2[1-\caE_r(\Psi)]}{1+s_0}, \\[1ex]
1-\frac{2\caE_r(\Psi)}{2+s_0+ s_1}\leq \frac{3}{2+s_0}\left[1-\frac{\caE_r(\Psi)}{1+\sqrt{s_0 s_1}}\right]-\frac{1-\caE_r(\Psi)-s_0}{2+s_0}\leq \frac{3}{2+s_0}\left[1-\frac{\caE_r(\Psi)}{1+\sqrt{s_0 s_1}}\right], \\[1ex]
\frac{2f_r(\Psi,E)+s_0+s_1}{2+s_0+s_1}\leq \frac{f_r(\Psi,E)+s_0}{1+s_0}\leq \frac{2f_r(\Psi,E)}{1+s_0},\\[1ex]
\frac{2f_r(\Psi,E)+s_0+s_1}{2+s_0+ s_1}\leq \frac{3}{2+s_0}\frac{f_r(\Psi,E)+\sqrt{s_0 s_1}}{1+\sqrt{s_0 s_1}}-\frac{f_r(\Psi,E)-s_0}{2+s_0}\leq \frac{3}{2+s_0}\frac{f_r(\Psi,E)+\sqrt{s_0 s_1}}{1+\sqrt{s_0 s_1}},
\end{gather}
which confirm the last inequalities in Eqs.~\eqref{eq:SepProbGBPSLUB}-\eqref{eq:SepProbErGBPSLUBH} and complete the proof of \thref{thm:SepProbGBPSLUB}.
\end{proof}

\subsection{Auxiliary lemmas}
Before proving \thsref{thm:SepProbGBPSmeanLUB} and \ref{thm:ConcSepProbforSNCinGBPS}, here we need to introduce two auxiliary lemmas. 

Denote by $S_{n-1}$ the $(n-1)$-dimensional unit sphere in $\bbR^n$. A function $f:S_{n-1}\rightarrow{} \bbR$ is a Lipschitz (continuous) function with Lipschitz constant $\eta$ if
\begin{equation}
	\left | f(x)-f(y) \right | \leq \eta\left \| x-y \right \|_2 \quad \forall \, x,y\in S_{n-1},
\end{equation}
where $\left \| \cdot \right \|_2$ denotes the Euclidean norm.

\begin{lemma}\label{lem:Levy'sLemma}
	\rm \textbf{(Levy's lemma)} \cite{Ledoux01,Watrous18} \it Suppose $f:S_{n-1}\rightarrow{} \bbR$ is a Lipschitz  function with Lipschitz constant $\eta$, and  $x\in S_{n-1}$ is drawn uniformly at random. Then 
	\begin{equation}\label{eq:Levy'sLemma}
		\Pr\left\{ f(x)-\bbE f \geq \epsilon  \right\} \leq 2\exp\left ( -\frac{n\epsilon^2}{25\pi \eta^2} \right ).
	\end{equation}
\end{lemma}

Next, define the following function on pure states in $\caH_{AB}$:
\begin{align}\label{eq:UrPsiDef}
	U_r(\Psi):=1-\frac{\caE_r(\Psi)}{1+s_0(\Psi)}, \quad |\Psi\>\in \caH_{AB}, \quad r\in [d-1],
\end{align}
which  can also be regarded as a function defined on a $(2D-1)$-dimensional real unit sphere, where $D=\dim(\caH_{AB})$.
Here $s_0(\Psi)$ is the largest Schmidt coefficient of $|\Psi\>$, recall that the Schmidt coefficients are arranged in nonincreasing order by default. The function $U_r(\Psi)$ is an upper bound for $P_\lc^\ub(\Psi,\caS_r)$ according to the following equation:
\begin{equation}\label{eq:UrPsileqPlcH}
	P_\lc^\ub(\Psi,\caS_r)=1-\frac{2\caE_r(\Psi)}{2+s_0(\Psi)+ s_1(\Psi)}\leq 1-\frac{\caE_r(\Psi)}{1+s_0(\Psi)}=U_r(\Psi),
\end{equation}
where the inequality holds because $s_0(\Psi)\geq s_1(\Psi)$. 
\begin{lemma}\label{lem:UrPsi}
	Suppose $|\Psi\>\in \caH_{AB}$ and  $r\in [d-1]$. Then $U_r(\Psi)$ is a Lipschitz function with Lipschitz constant 2.
	If $|\Psi\>$ is a Haar-random pure state in $\caH_{AB}$, then 
	\begin{equation}\label{eq:UrPsiUB} 
		\bbE U_r(\Psi)< \frac{4(r+1)}{d+1}. 
	\end{equation}
\end{lemma}

\begin{proof}[Proof of \lref{lem:UrPsi}]
	Let $|\Upsilon\>\in \caH_{AB}$ be another pure state; let $s_0(\Psi)$ and  $s_0(\Upsilon)$ be the largest Schmidt coefficients of $|\Psi\>$ and $|\Upsilon\>$, respectively. Then we have
	\begin{align}
		&U_r(\Upsilon)-	U_r(\Psi)=\frac{\caE_r(\Psi)}{1+s_0(\Psi)}-\frac{\caE_r(\Upsilon)}{1+s_0(\Upsilon)}=\frac{\caE_r(\Psi)}{1+s_0(\Psi)}-\frac{\caE_r(\Upsilon)}{1+s_0(\Psi)}+\frac{\caE_r(\Upsilon)}{1+s_0(\Psi)}-\frac{\caE_r(\Upsilon)}{1+s_0(\Upsilon)}\nonumber\\
	&=\frac{\caE_r(\Psi)-\caE_r(\Upsilon)}{1+s_0(\Psi)}+\frac{\caE_r(\Upsilon)[s_0(\Upsilon)-s_0(\Psi)]}{[1+s_0(\Psi)][1+s_0(\Upsilon)]}\leq  \frac{\||\Psi\>-|\Upsilon\>\|_2}{1+s_0(\Psi)}+\frac{\caE_r(\Upsilon)\||\Psi\>-|\Upsilon\>\|_2}{[1+s_0(\Psi)][1+s_0(\Upsilon)]} \leq 2 \||\Psi\>-|\Upsilon\>\|_2,
		\end{align}
	where the inequalities hold because  $s_0(\Psi)=1-\caE_1(\Psi)$, $0\leq \caE_r(\Psi)\leq 1$, and $\caE_r(\Psi)$ is a Lipschitz function with Lipschitz constant 1 by \psref{pro:ErLUB} and \ref{pro:ErLip}. By symmetry the above equation still holds if $|\Psi\>$ and $|\Upsilon\>$ are exchanged. Therefore,  $U_r(\Psi)$ is a Lipschitz function with Lipschitz constant 2.

	Next, we turn to \eref{eq:UrPsiUB}, assuming that  $|\Psi\>$ is a Haar-random pure state in $\caH_{AB}$. Let $\{s_j=s_j(\Psi)\}_{j=0}^{d-1}$ denote the Schmidt spectrum of $|\Psi\>$. By definition we have
	\begin{align}
		U_r(\Psi) =  1-  \frac{\caE_r(\Psi)}{1+s_0}=1-  \frac{\sum_{j=r}^{d-1} s_j}{1+s_0}
		=\frac{\sum_{j=0}^{r-1} s_j}{1+s_0}
		+\frac{s_0}{1+s_0}\leq \frac{d}{d+1}\sum_{j=0}^{r-1} s_j+\frac{s_0}{1+s_0},
	\end{align}
	where the inequality holds because $s_0\geq 1/d$.
	Therefore, 
	\begin{align} 
		\bbE U_r(\Psi) &\leq \frac{d}{d+1} \bbE \sum_{j=0}^{r-1} s_j
		+\bbE \frac{s_0}{1+s_0}< \frac{4r}{d+1}+\frac{4}{d+1} = \frac{4(r+1)}{d+1},
	\end{align}
	which confirms \eref{eq:UrPsiUB}. Here the second inequality holds because
	\begin{align}
		\bbE \sum_{j=0}^{r-1} s_j\leq r \bbE s_0\leq \frac{4r}{d},\quad \bbE \frac{s_0}{1+s_0}\leq \frac{\bbE s_0}{1+\bbE s_0}\leq \frac{4}{d+4}< \frac{4}{d+1},
	\end{align}
	given that $\bbE s_0\leq 4/d$ by Lemma 22 in \rcite{Liu18} and that the function $s_0/(1+s_0)$ is concave in $s_0$. 
\end{proof}

\subsection{Proof of \thref{thm:SepProbGBPSmeanLUB}}

\begin{proof}[Proof of \thref{thm:SepProbGBPSmeanLUB}]
	The first inequality in \eref{eq:SepProbGBPSmeanLUB} follows from \thref{thm:SepProbforEMCinMES} given that $P_\sep(\Psi,\caS_r)\geq P_\sep(\Phi,\caS_r)$  for all $|\Psi\>\in \caH_{AB}$ by \thref{thm:SepProb&Major}. The second, third, and fourth inequalities in \eref{eq:SepProbGBPSmeanLUB} hold by definition. 
	The last inequality in \eref{eq:SepProbGBPSmeanLUB} holds because $P_\lc^\ub(\Psi,\caS_r)\leq U_r(\Psi)$ by \eref{eq:UrPsileqPlcH} and $\bbE U_r(\Psi)< 4(r+1)/(d+1)$ by \lref{lem:UrPsi}. 
\end{proof}

\subsection{Proof of \thref{thm:ConcSepProbforSNCinGBPS}}

\begin{proof}[Proof of \thref{thm:ConcSepProbforSNCinGBPS}]
	\Eref{eq:ConcSepProbforSNCinGBPS} in \thref{thm:ConcSepProbforSNCinGBPS} can be proved as follows:
	\begin{equation}
	\!\!	\Pr\left\{ P_\lc^\rmH(\Psi,\caS_r) \geq \frac{4(r+1)}{d+1}+\epsilon \right\} \leq \Pr\left\{ U_r(\Psi)\geq \frac{4(r+1)}{d+1}+\epsilon \right\} 
		\leq \Pr\left\{ U_r(\Psi)\geq \bbE  U_r(\Psi)+\epsilon \right\} 		
		\leq 2\exp\left ( -\frac{D\epsilon^2}{50\pi} \right) .
	\end{equation}
	Here the first two inequalities follow from the facts that $P_\lc^\rmH(\Psi,\caS_r)\leq P_\lc^\ub(\Psi,\caS_r)\leq U_r(\Psi)$ by \eref{eq:UrPsileqPlcH} and that $\bbE U_r(\Psi)< 4(r+1)/(d+1)$ by \lref{lem:UrPsi}; the last inequality follows from \lref{lem:Levy'sLemma} (Levy's lemma) with $n=2D$ and $\eta=2$ given that $U_r(\Psi)$ is a Lipschitz function with Lipschitz constant 2 by \lref{lem:UrPsi} again.

	By definition we have $P_\sep(\Psi,\caS_r)\leq P_\lc(\Psi,\caS_r)\leq P_\lc^\rmH(\Psi,\caS_r)$ and $P_\sep(\Psi,\caS_r)\leq P_\sep^\rmH(\Psi,\caS_r)\leq P_\lc^\rmH(\Psi,\caS_r)$, so 
	\eref{eq:ConcSepProbforSNCinGBPS} still holds if $P_\lc^\rmH(\Psi,\caS_r)$ is replaced by $P_\lc(\Psi,\caS_r)$, $P_\sep^\rmH(\Psi,\caS_r)$, or $P_\sep(\Psi,\caS_r)$, which completes the proof of \thref{thm:ConcSepProbforSNCinGBPS}. 
\end{proof}

\section{\label{app:OS&SPfor2qPS}Proofs of results on two-qubit pure states}

In this appendix we prove \lsref{lem:OptStratforECin2qPS}-\ref{lem:POmegathetap2}  and \pref{pro:OmegathetapLC}, which are tied to  entanglement certification of two-qubit pure states.

\subsection{Proof of \lref{lem:OptStratforECin2qPS}}

\begin{proof}[Proof of \lref{lem:OptStratforECin2qPS}]
	Recall that the target state $|\Psi_\theta\>$ is invariant under  swap, complex conjugation (with respect to the computational basis), and any unitary transformation of the form $V_\zeta\otimes V_\zeta^*$, where $V_\zeta=|0\>\<0|+\rme^{-\rmi \zeta}|1\>\<1|$ and $0\leq \zeta< 2\pi$. Therefore, according to \pref{pro:OptOmegaSym}, we can restrict our attention to verification operators that enjoy the same symmetry when searching for an optimal verification operator. Following a similar analysis presented in \rcite{PLM18} we can deduce that any such verification operator has the following form: 
	\begin{align}\label{eq:Omegalambda23}
		\Omega=|\Psi_\theta\>\<\Psi_\theta|+\lambda_2 |\Psi_\theta^\perp\>\<\Psi_\theta^\perp|+\lambda_3(|01\>\<01|+|10\>\<10|), \quad 0\leq \lambda_2, \lambda_3\leq 1,
	\end{align}
	where $|\Psi_\theta^\perp\>=\sin\theta|00\>-\cos\theta|11\>$ is orthogonal to the target state $|\Psi_\theta\>$.
	
		\begin{figure}[b]
		\centering
		\includegraphics[scale=0.58]{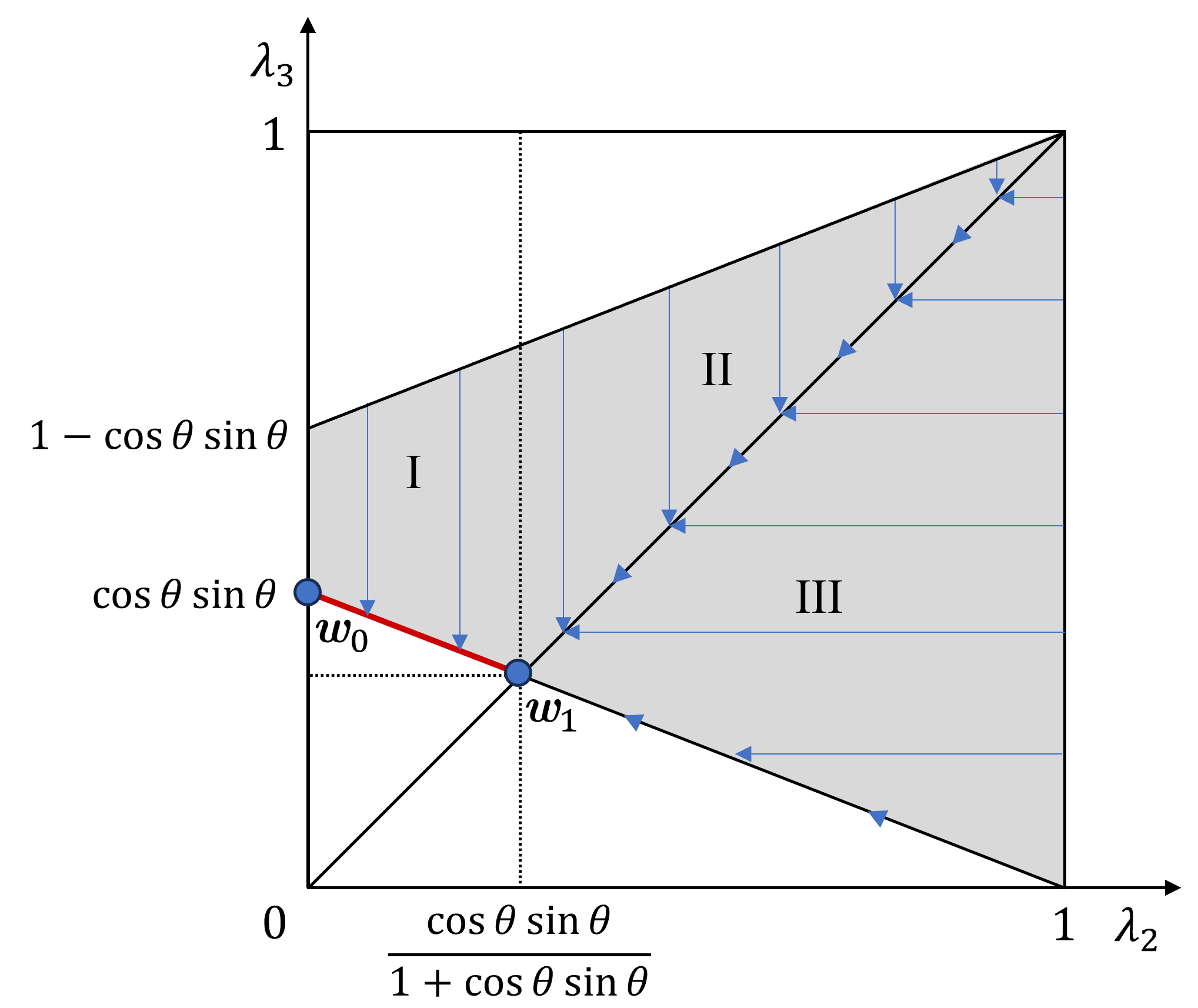}
		\caption{\label{fig:lambda}  Verification operators of $|\Psi_\theta\>$ that have the form  \eref{eq:Omegalambda23}. Separable 
			verification operators correspond to points in the shaded region, which is constrained by \eref{eq:LambdaConstraint1}.  The blue arrows indicate optimization directions along which the separation probability is nonincreasing. One point on the red line segment $(\bmw_0,\bmw_1)$ corresponds to an optimal verification operator. }
	\end{figure}

	Now, suppose $\Omega$ is associated with a separable verification strategy, then both $\Omega$ and $\bbone-\Omega$ are separable operators. In the case of two qubits under consideration, a positive operator is separable iff it is positive after partial transpose with respect to either party, say Bob. Simple calculation shows that  $\Omega^{T_B}$ has the following four eigenvalues: 
	\begin{equation}
		\cos^2\theta+\lambda_2\sin^2\theta, \quad \sin^2\theta+\lambda_2\cos^2\theta, \quad \lambda_3+(1-\lambda_2)\cos\theta\sin\theta, \quad \lambda_3-(1-\lambda_2)\cos\theta\sin\theta.
	\end{equation}
	The verification strategy $\Omega$ is  separable iff all these eigenvalues lie in the interval $[0,1]$, which amounts to the following conditions:
	\begin{equation}\label{eq:LambdaConstraint1}
		\lambda_3+(1-\lambda_2)\cos\theta\sin\theta\leq 1, \quad \lambda_3-(1-\lambda_2)\cos\theta\sin\theta\geq 0,
	\end{equation}
	as illustrated in \fref{fig:lambda}. 
	Note that the separation probability $P_\Omega(\Psi_\theta)$ is nondecreasing in $\lambda_2$ and $\lambda_3$; in addition,	
	the first inequality is satisfied automatically when the second inequality is saturated. To construct an optimal separable verification operator, it suffices to consider the case in which the second inequality is saturated, that is,  
	\begin{equation}
		\lambda_2+\frac{\lambda_3}{\cos\theta\sin\theta}= 1.
	\end{equation}

	If $\lambda_2\geq \lambda_3$, then $\beta(\Omega)=\lambda_2$ and $\nu(\Omega)=1-\lambda_2$, which means $\tr(\Omega\sigma)\leq \cos^2\theta+\lambda_2\sin^2\theta$ for any $\sigma\in \caS_\sep$ by \pref{pro:PassProbUB}, and the upper bound can be attained when  $\sigma=|00\>\<00|$. So the separation probability $P_\Omega(\Psi_\theta)$ reads
	\begin{equation}
		P_\Omega(\Psi_\theta)=\max_{\sigma\in\caS_\sep}\tr(\Omega\sigma)=\cos^2\theta+\lambda_2\sin^2\theta, 
	\end{equation}
	which increases monotonically with $\lambda_2$, assuming that $\lambda_2\geq \lambda_3$. To construct an optimal separable verification operator, therefore, it suffices to consider the parameter range  defined by the following conditions:
	\begin{equation}\label{eq:LambdaConstraint2}
		\lambda_2+\frac{\lambda_3}{\cos\theta\sin\theta}= 1,\quad 0\leq \lambda_2 \leq \lambda_3,
	\end{equation}
	which corresponds to the red line segment in \fref{fig:lambda}. The two endpoints of this line segment read
	\begin{equation}
		\bmw_0:=\left(0,\cos\theta\sin\theta\right),\quad 		\bmw_1:=\left(\frac{\cos\theta\sin\theta}{1+\cos\theta\sin\theta},\frac{\cos\theta\sin\theta}{1+\cos\theta\sin\theta}\right),
	\end{equation}
	which correspond to the verification operators $\Omega_0$ and $\Omega_1$ defined in \eqsref{eq:Omega0}{eq:Omega1}, respectively. It follows that an optimal separable verification operator can be constructed from a convex combination of $\Omega_0$ and $\Omega_1$. In other words, for some $p\in [0,1]$, $\Omega(\theta,p)$ is an optimal separable verification operator of the target state $|\Psi_\theta\>$, which amounts to  \eref{eq:SepProb2qubitthetap}. This observation completes the proof of  \lref{lem:OptStratforECin2qPS}.
\end{proof}

\subsection{Proof of \lref{lem:POmegathetap}}

\begin{proof}[Proof of \lref{lem:POmegathetap}]
	Let 
	\begin{align}
		|\psi_{a,\xi }\>=\cos a|0\>+\sin a \,\rme^{\rmi\xi}|1\>, \quad \rho_{a,\xi}=|\psi_{a,\xi }\>\<\psi_{a,\xi}|, \quad 0\leq a\leq \pi/2, \quad 0\leq \xi<2\pi. 
	\end{align}
	When $\xi=0$, $|\psi_{a,\xi }\>$ and  $\rho_{a,\xi}$ can be abbreviated as   $|\psi_{a}\>$ and $\rho_{a}$, respectively. Then  $P(\theta,p)$ can  be expressed as follows:
	\begin{align}\label{eq:POmegathetapProof}
		P(\theta,p)&=\max_{0\leq a,b\leq \pi/2, \, 0\leq \xi_1, \xi_2<2\pi} \tr\left[\Omega(\theta,p)\rho_{a,\xi_1 }\otimes \rho_{b,\xi_2}\right]=\max_{0\leq a,b\leq \pi/2} \tr[\Omega(\theta,p)(\rho_a\otimes \rho_b)],
	\end{align}
	where the second equality holds because all entries of $\Omega(\theta,p)$ in the computational basis are nonnegative. In addition, direct calculation yields
	\begin{align}
		\tr[\Omega(\theta,p)(\rho_a\otimes \rho_b)]&=\left(\cos\theta\cos a\cos b+\sin\theta\sin a\sin b\right)^2+\frac{p\cos\theta\sin\theta}{1+\cos\theta\sin\theta}\left(\sin\theta\cos a\cos b-\cos\theta\sin a\sin b\right)^2 \notag \\
		&\equad+\cos\theta\sin\theta\left(1-\frac{p\cos\theta\sin\theta}{1+\cos\theta\sin\theta}\right)\left(\cos^2 a\sin^2 b+\sin^2 a\cos^2 b\right).
	\end{align}

	Let $u=(\cos^2 a+\cos^2 b)/2$ and $v=(\cos^2 a-\cos^2 b)/2$. Then $\tr[\Omega(\theta,p)(\rho_a\otimes \rho_b)]$ can also be expressed as follows:
	\begin{align}
		\tr[\Omega(\theta,p)(\rho_a\otimes \rho_b)]&=\left [\cos\theta\sqrt{(u+v)(u-v)}+\sin\theta\sqrt{\left (1-u-v\right )\left (1-u+v\right )}\,\right ]^2 \notag \\
		&\equad+\frac{p\cos\theta\sin\theta}{1+\cos\theta\sin\theta} \left [\sin\theta\sqrt{(u+v)(u-v)}-\cos\theta\sqrt{\left (1-u-v\right )\left (1-u+v\right )}\,\right ]^2 \notag \\
		&\equad+2\cos\theta\sin\theta\left(1-\frac{p\cos\theta\sin\theta}{1+\cos\theta\sin\theta}\right)\left [u-(u+v)(u-v)\right ].
	\end{align}
	Its partial derivative over $v$ reads
	\begin{equation}\label{eq:gamma}
		\frac{\partial\tr\left[\Omega(\theta,p)(\rho_a\otimes \rho_b)\right]}{\partial v}=-2\gamma v,
	\end{equation}
	where
	\begin{align}
		\gamma&=1+\cos\theta\sin\theta\left(\tan a\tan b+\cot a\cot b\right)+\frac{p\cos\theta\sin\theta}{1+\cos\theta\sin\theta}\left[1-\cos\theta\sin\theta(\tan a\tan b+\cot a\cot b)\right] \notag \\
		&\equad-2\cos\theta\sin\theta\left(1-\frac{p\cos\theta\sin\theta}{1+\cos\theta\sin\theta}\right)\geq 1+\frac{p\cos\theta\sin\theta}{1+\cos\theta\sin\theta}\geq 1.
	\end{align}
	Therefore, $\tr[\Omega(\theta,p)(\rho_a\otimes \rho_b)]$ is strictly increasing in $v$ when $v\leq 0$, and strictly decreasing in $v$ when $v\geq 0$. To evaluate the last  maximization in \eref{eq:POmegathetapProof}, we can take $v=0$, that is, $b=a$, which confirms the first equality in \eref{eq:POmegathetap}.

	Now, direct calculation yields
	\begin{align}
		\tr\left[\Omega(\theta,p)\rho_a^{\otimes 2}\right]&=\left (\cos\theta\cos^2 a+\sin\theta\sin^2 a\right )^2+2\cos\theta\sin\theta\cos^2 a\sin^2 a \notag \\
		&\equad+\frac{p\cos\theta\sin\theta}{1+\cos\theta\sin\theta}\left [ \left (\sin\theta\cos^2 a-\cos\theta\sin^2 a\right )^2-2\cos\theta\sin\theta\cos^2 a\sin^2 a\right ].
	\end{align}
	To determine the maximum of $\tr\left[\Omega(\theta,p)\rho_a^{\otimes 2}\right]$ over $a$, we can take its partial derivative with respect to $a$:
	\begin{align}\label{eq:POmegathetaPartialDerivative}
		\frac{\partial \tr\left[\Omega(\theta,p)\rho_a^{\otimes 2}\right]}{\partial a}&=4\cos a\sin a\left[g_1(\theta,p)\cos^2 a -g_2(\theta,p)\sin^2 a\right],
	\end{align}
	where	
	\begin{equation}
		\begin{aligned}
			g_1(\theta,p)&=2
			\cos\theta\sin\theta-\cos^2\theta-\frac{p\cos\theta\sin\theta}{1+\cos\theta\sin\theta}\left(2\cos\theta\sin\theta+\sin^2\theta\right), \\
			g_2(\theta,p)&=2\cos\theta\sin\theta-\sin^2\theta-\frac{p\cos\theta\sin\theta}{1+\cos\theta\sin\theta}\left(2\cos\theta\sin\theta+\cos^2\theta\right).
		\end{aligned}
	\end{equation}
	In conjunction with the assumptions   $0< \theta\leq \pi/4$ and $0\leq p\leq1$ we can deduce  that
	\begin{equation}
		g_2(\theta,p)\geq g_2(\theta,1)=\frac{\cos(2\theta)+\sin(2\theta)-1}{2+\sin(2\theta)}\geq 0,\quad 
		g_2(\theta,p)-g_1(\theta,p)=\frac{\cos(2\theta)[2+(1-p)\sin(2\theta)]}{2+\sin(2\theta)}\geq 0.
	\end{equation}
	Here the first inequality is saturated iff $p=1$,  the second inequality is saturated iff $\theta=\pi/4$, and the third inequality is saturated iff $\theta=\pi/4$.

	If $q(\theta)\leq p\leq 1$, where $q(\theta)$ is defined in \eref{eq:thetaFun} and satisfies $q(\theta)\leq 1$, then $g_1(\theta,p)\leq 0$, so $\tr\left[\Omega(\theta,p)\rho_a^{\otimes 2}\right]$ is nonincreasing in $a$ and is thus maximized when  $a=0=a^*(\theta,p)$. Therefore,
	\begin{equation}\label{eq:PthetapProof1}
		P(\theta,p)=\tr\left[\Omega(\theta,p)|00\>\<00|\right]=\cos^2\theta+\frac{p\cos\theta\sin^3\theta}{1+\cos\theta\sin\theta},
	\end{equation}
	which confirms  the second equality in \eref{eq:POmegathetap} and shows that  $P(\theta,p)$ is strictly increasing in $p$; this result also confirms  the third equality in \eref{eq:POmegathetap} except when $\theta=\pi/4$.  If in addition $p<1$ or $\theta<\pi/4$, then $g_2(\theta,p)> 0$, so $\tr\left[\Omega(\theta,p)\rho_a^{\otimes 2}\right]$ is strictly decreasing in $a$, and its maximum over $a\in [0,\pi/2]$ is attained iff $a=a^*(\theta,p)$.

	If instead  $0\leq p< q(\theta)$, then $p<1$, $g_2(\theta,p)>0$,  $0\leq g_1(\theta,p)\leq g_2(\theta,p)$, and $g_1(\theta,p)/g_2(\theta,p)=h(\theta,p)$, where $h(\theta,p)$ is defined in \eref{eq:thetaFun}. In addition, the function $g_1(\theta,p)\cos^2 a -g_2(\theta,p)\sin^2 a$ is strictly  decreasing in $a$ for $a\in [0,\pi/2]$ and is equal to 0 when $a= \arctan\sqrt{h(\theta,p)}
	=a^*(\theta,p)$ [see \eref{eq:thetaFun}]. Therefore, the maximum of $\tr\left[\Omega(\theta,p)\rho_a^{\otimes 2}\right]$ over $a\in [0,\pi/2]$ is attained iff  $a=\arctan\sqrt{h(\theta,p)}=a^*(\theta,p)$, which confirms the second equality in \eref{eq:POmegathetap}.
	
	Next, we prove the third  equality in \eref{eq:POmegathetap} when $\theta=\pi/4$. In this case,  $q(\theta)=1$, $h(\theta,p)=1$ for $p\in [0,1)$, and 
	\begin{equation}
		a^*(\theta,p)=\begin{cases}
			\pi/4 & p\in [0,1),\\
			0 &p=1,
		\end{cases}
	\end{equation}
	so $\tr\left[\Omega(\theta,p)\rho_{a^*}^{\otimes 2}\right]=(9-p)/12$,  which confirms the third  equality in \eref{eq:POmegathetap}.

	Finally, we turn to the convexity of the separation probability $P(\theta,p)$, assuming that $\arctan(1/2) < \theta<\pi/4$ and $0\leq p\leq q(\theta)$. Then  the maximum of $\tr\left[\Omega(\theta,p)\rho_a^{\otimes 2}\right]$ over $a\in [0,\pi/2]$ is attained iff  $a=\arctan\sqrt{h(\theta,p)}=a^*(\theta,p)$. In addition, according to \eref{eq:thetaFun} and the following equation
	\begin{equation}
		\frac{\partial h(\theta,p)}{\partial p}=\frac{(\sin^4\theta-\cos^4\theta)}{g_2^2(\theta,p)(1+\cos\theta\sin\theta)}<0,
	\end{equation}
$h(\theta,p)$ and $a^*(\theta,p)$ are strictly decreasing in $p$. Suppose $0\leq p_1<p_2 \leq  q(\theta)$, $0<x<1$, and $p=xp_1+(1-x)p_2$; then
	\begin{align}
		P(\theta,p)&=\tr\left[\Omega(\theta,p)\rho_{a^*(\theta,p)}^{\otimes 2}\right]=x\tr\left[\Omega(\theta,p_1)\rho_{a^*(\theta,p)}^{\otimes 2}\right]+(1-x) \tr\left[\Omega(\theta,p_2)\rho_{a^*(\theta,p)}^{\otimes 2}\right]  \nonumber \\
		&<x\tr\left[\Omega(\theta,p_1)\rho_{a^*(\theta,p_1)}^{\otimes 2}\right]+(1-x) \tr\left[\Omega(\theta,p_2)\rho_{a^*(\theta,p_2)}^{\otimes 2}\right]=x P(\theta,p_1)+(1-x)P(\theta,p_2). 
	\end{align}
	So   $P(\theta,p)$ is strictly convex in $p$ for $p\in [0,q(\theta)]$, which  completes the proof of \lref{lem:POmegathetap}.
\end{proof}

\subsection{Proof of \lref{lem:POmegathetap2}}
\begin{proof}[Proof of \lref{lem:POmegathetap2}]
	By definition we have $P(\theta,p)=P_{\Omega(\theta,p)}(\Psi_\theta)=\max_{\sigma\in \caS_\sep} \tr[\Omega(\theta,p)\sigma]$, so 
	$P(\theta,p)$ is convex in $p$ given that $\Omega(\theta,p)$ is linear in $p$ by construction.

	If $0< \theta \leq \arctan(1/2)$, then $q(\theta)=0$ and  $P(\theta,p)$ is strictly increasing in $p$ by \lref{lem:POmegathetap} and thus has a unique minimizer at $p=0$, that is $p^*(\theta)=0$, which confirms \eref{eq:p*} for $\theta\in (0,\arctan(1/2)]$.  If $\theta=\pi/4$, then $q(\theta)=1$ and $P(\theta,p)=(9-p)/12$  is strictly decreasing in $p$ and thus has a unique minimizer at $p=1$, that is $p^*(\theta)=1$, which confirms \eref{eq:p*} again. In both cases, we have $p^*(\theta)\in [0,q(\theta)]$.

	Next, suppose $\arctan(1/2)< \theta < \pi/4$; then $0<q(\theta)< 1$. According to \lref{lem:POmegathetap}, $P(\theta,p)$ is strictly increasing in $p$ for $p\in [q(\theta),1]$ and  strictly convex in $p$ for $p\in [0,q(\theta)]$. Therefore, $P(\theta,p)$  has a unique minimizer and   $p^*(\theta)\in [0,q(\theta)]$; in addition,  $P(\theta,p)$ is strictly decreasing (increasing) in $p$ for $p\in [0,p^*(\theta)]$ $(p\in [p^*(\theta),1])$.
	In conjunction with \eref{eq:thetaFun} we can further deduce that
	\begin{align}
		P_p(\theta, 0)	&=\frac{\sin(2\theta)[17 - 9 \cos(4 \theta) - 25 \sin(2 \theta) + 3\sin(6 \theta)]}{8[2\sin(2\theta)-1]^2[2+\sin(2\theta)]},\\ 
		P_p(\theta, p)	&=\frac{\cos\theta\sin^3\theta}{1+\cos\theta\sin\theta}>0 \quad \forall\, p\in [q(\theta),1],
	\end{align}
	where $P_p(\theta, p)$ is a shorthand for $\partial P(\theta, p)/\partial p$. Note that $P_p(\theta, p)$ is continuous in $p$ for $p\in [0,1]$.

	If  $\arctan(1/2)< \theta \leq \theta^*$, then $P_p(\theta, p)\geq P_p(\theta, 0)\geq 0$ for $p\in [0,1]$; note that $P_p(\theta^*, 0)=0$ according to the definition of $\theta^*$ based on \eref{eq:theta*}. 
	Therefore, $P(\theta,p)$ is strictly increasing in $p$ for $p\in [0,1]$, which means  $p^*(\theta)=0$ and  confirms \eref{eq:p*} for $\theta\in (0,\theta^*]$ given the above analysis.

	If  $\theta^*<\theta < \pi/4$, then $ P_p(\theta, p=0)< 0$, while $P_p(\theta, p=q(\theta))>0$. So $P_p(\theta, p)=0$ has a unique zero for $p\in (0,q(\theta))$, which necessarily coincides with the minimizer $p^*(\theta)$. This observation  completes the proof of \lref{lem:POmegathetap2}.
\end{proof}

\subsection{\label{app:OmegaWH}Proof of \pref{pro:OmegathetapLC}}

To prove \pref{pro:OmegathetapLC}, we need to introduce a verification strategy for $|\Psi_\theta\>$ constructed by Wang and Hayashi~\cite{WangH19}, assuming that $0<\theta\leq \pi/4$. First, \rcite{WangH19} constructed the following five test operators for $|\Psi_\theta\>$ using LOCC: 
\begin{equation}
	\begin{aligned}
		T_1^{A\to B}&=\eta |0\>\<0|\otimes |0\>\<0|+|\tpsi_+\>\<\tpsi_+|\otimes |+\>\<+|+|\tpsi_-\>\<\tpsi_-|\otimes |-\>\<-| ,\\
		T_2^{A\to B}&=\eta |0\>\<0|\otimes |0\>\<0|+|\tvarphi_+\>\<\tvarphi_+|\otimes |\top\>\<\top|+|\tvarphi_-\>\<\tvarphi_-|\otimes |\perp\>\<\perp| ,\\
		T_1^{B\to A}&=\eta |0\>\<0|\otimes |0\>\<0|+|+\>\<+|\otimes |\tpsi_+\>\<\tpsi_+|+|-\>\<-|\otimes |\tpsi_-\>\<\tpsi_-| ,\\
		T_2^{B\to A}&=\eta |0\>\<0|\otimes |0\>\<0|+|\top\>\<\top|\otimes |\tvarphi_+\>\<\tvarphi_+|+|\perp\>\<\perp|\otimes |\tvarphi_-\>\<\tvarphi_-| ,\\
		T_3&=|0\>\<0|\otimes |0\>\<0|+|1\>\<1|\otimes |1\>\<1|,
	\end{aligned}
\end{equation}
where
\begin{equation}
	\begin{gathered}
		|+\>=\frac{1}{\sqrt{2}}(|0\>+|1\>), \quad |-\>=\frac{1}{\sqrt{2}}(|0\>-|1\>), \quad |\top\>=\frac{1}{\sqrt{2}}(|0\>+\rmi|1\>), \quad |\perp\>=\frac{1}{\sqrt{2}}(|0\>-\rmi|1\>) ,\\
		|\tpsi_\pm\>=\frac{(1-\eta)\cos\theta}{\sqrt{1-\eta\cos^2\theta}}|0\> \pm \frac{\sin\theta}{\sqrt{1-\eta\cos^2\theta}}|1\> , \quad |\tvarphi_\pm\>=\frac{(1-\eta)\cos\theta}{\sqrt{1-\eta\cos^2\theta}}|0\> \pm \rmi\frac{\sin\theta}{\sqrt{1-\eta\cos^2\theta}}|1\>.
	\end{gathered}
\end{equation}
Note that $|\tpsi_\pm\>$ and $|\tvarphi_\pm\>$ are not normalized. Based on these test operators, \rcite{WangH19} constructed the following verification strategy for $|\Psi_\theta\>$:
\begin{equation}
	\Omega_{\mathrm{WH}}(\theta,\eta,p')=\frac{1-p'}{4}\left(T_1^{A\to B}+T_2^{A\to B}+T_1^{B\to A}+T_2^{B\to A}\right)+p'T_3,
\end{equation}
which can be realized by LOCC.

Now, suppose $\tp(\theta)\leq p\leq 1$, where $\tp(\theta)$ is defined in \eref{eq:tptheta}. Let
\begin{equation}\label{eq:tildepplinear}
	\eta=1-\tan\theta,\quad 	p'=\cos\theta\sin\theta\left(\frac{\cos^2\theta+\cos\theta\sin\theta}{1+\cos\theta\sin\theta}p+\tan\theta-1\right);
\end{equation}
then $0\leq p'\leq \sin^2\theta/(1+\cos\theta\sin\theta)\leq 1/3$, and it is straightforward to verify the following equality:
\begin{equation}
	\Omega(\theta,p)=\Omega_{\mathrm{WH}}(\theta,\eta,p').
\end{equation}
So  the verification strategy $\Omega(\theta,p)$ can be realized by LOCC, which completes the proof of \pref{pro:OmegathetapLC}.

\end{document}